\documentclass[AMS, STIX1COL]{WileyNJD-v2}

\articletype{Article Type}%

\usepackage{amsmath,amssymb,stmaryrd}
\usepackage{amsthm}
\usepackage{enumerate}
\usepackage{url}
\usepackage{layouts}
\usepackage{color}

\usepackage{mathtools}
\DeclarePairedDelimiter\ceil{\lceil}{\rceil}
\DeclarePairedDelimiter\abs{\lvert}{\rvert}

\usepackage{caption}
\usepackage{subcaption}
\usepackage{graphicx}

\usepackage{epstopdf}

\usepackage{setspace}

\DeclarePairedDelimiter\floor{\lfloor}{\rfloor}

\received{26 April 2016}
\revised{6 June 2016}
\accepted{6 June 2016}

\begin{document}

\title{Space-efficient estimation of empirical tail dependence coefficients for bivariate data streams}

\author[1,2]{Alastair Gregory*}

\author[1,2]{Kaushik Jana}

\authormark{Alastair Gregory \textsc{et al}}

\address[1]{\orgdiv{Programme for Data-Centric Engineering}, \orgname{Alan Turing Institute}, \orgaddress{\state{London}, \country{U.K.}}}

\address[2]{\orgdiv{Department of Mathematics}, \orgname{Imperial College London}, \orgaddress{\state{London}, \country{U.K.}}}

\corres{*Dr Alastair Gregory, Department of Mathematics, Imperial College London, SW7 2AZ. \email{a.gregory14@imperial.ac.uk}}

\presentaddress{This is sample for present address text this is sample for present address text}

\abstract[Summary]{
This article proposes a space-efficient approximation to empirical tail dependence coefficients of an indefinite bivariate stream of data. The approximation, which has stream-length invariant error bounds, utilises recent work on the development of a summary for bivariate empirical copula functions. The work in this paper accurately approximates a bivariate empirical copula in the tails of each marginal distribution, therefore modelling the tail dependence between the two variables observed in the data stream. Copulas evaluated at these marginal tails can be used to estimate the tail dependence coefficients. Modifications to the space-efficient bivariate copula approximation, presented in this paper, allow the error of approximations to the tail dependence coefficients to remain stream-length invariant. Theoretical and numerical evidence of this, including a case-study using the Los Alamos National Laboratory netflow data-set, is provided within this article.}


\jnlcitation{\cname{%
\author{A. Gregory}, and \author{K. Jana}} (\cyear{2019}), 
\ctitle{Space-efficient estimation of empirical tail dependence coefficients for bivariate data streams}}

\maketitle


\section{Introduction}

Streaming data arises in a wide array of contemporary research fields, including but not limited to the Internet of Things (IoT), cyber-security \citep{Adams2016} and continuous sensor observation \citep{Golab}. The data is acquired continuously and usually at a fast pace. Typically one can't store all of the data ever streamed due to memory constraints and it is infeasible to repeat statistical analyzes on the entire stream as it grows indefinitely over time. To deal with this problematic situation, estimation methods for many different statistical analyzes on such data streams have been proposed \citep{Gama, Aggarwal}. Some of these techniques use statistical summaries, where only a succinct number of carefully selected elements that have entered into the data stream are stored. One of the most common analyzes that the literature has focused on is quantile estimation for a univariate variable observed in a data stream (including median estimation) \citep{Buragohain}. Recently, \cite{Gregory} proposed an algorithm to generate an approximation to the bivariate empirical copula function (a common method of nonparametric dependence modelling) using a succinct statistical summary; this approximation has a guaranteed error bound.

This paper builds on the work in \cite{Gregory} and considers an important use of copulas: computing tail dependence coefficients between random variables \citep{Sibuya}. Tail dependence coefficients between random variables quantify their co-movement in the tails of the marginals. For example, two random variables may be weakly dependent in the vast majority of their probability space, however in the tails of this space they may be highly dependent. This behaviour is often seen in financial analysis \citep{Rodriguez}, where sometimes two assets exhibit sharp price increases and decreases at similar times, but tend to have relatively uncorrelated typical daily price movements. The importance of tail dependence in fields such as hydrology \citep{Poulin} and energy \citep{Reboredo} has also been studied. For the purpose of approximating empirical tail dependence coefficients between streams of data, with stream-length invariant error, the type of summary that was proposed in \cite{Gregory} is not sufficient since it returned uniform error over the marginal distributions.
Therefore the copula summary in \cite{Gregory} results in an approximation to the tail dependence coefficient with linearly growing error w.r.t. the number of elements in the data stream. To remedy this, we propose to use relative accuracy quantile summaries from the univariate literature \citep{Cormode} within the copula summary, which allows suitable properties of the modified summaries' approximation error to be proved. These properties lead to approximations of the tail dependence coefficients that have constant error w.r.t. the number of elements in the data stream.

This article is structured as follows. The next section introduces the empirical copula and how it can be used to construct empirical tail dependence coefficients between random variables observed through data. This section also describes the challenges associated with computing empirical copula approximations when the data is streamed sequentially, and how the copula summary proposed in \cite{Gregory} can be used to provide an approximation to copulas in this streaming regime. Then Sec. \ref{sec:estimationtaildependence} introduces how one can adapt the summary to compute accurate approximations to empirical tail dependence coefficients for streaming data. Sections \ref{sec:theoreticalanalysis} and \ref{sec:numericalanalysis} provide a theoretical and numerical analysis of the approximation respectively. In particular, Sec. \ref{data_analysis} provides a case-study of applying the proposed methodology to the Los Alamos National Laboratory (LANL) netflow data-set.

\section{Bivariate empirical copulas and tail dependence}

\label{sec:empiricalcopulas}

A copula is a dependence model between two or more different random variables \citep{Sklar}. The work in \cite{Gregory} explains how one can consider higher dimensional (greater than two) copulas in the context of streaming data (the motivation of this paper), using decompositions involving pair-wise copula models \citep{Aas}. This is described further in Sec. \ref{sec:spacememory}. However, this aspect is outside the scope of the current study on tail dependence, and therefore for the remainder of this paper we will focus on the case where there is only two random variables. 
More specifically, a bivariate copula function $C(u_1,u_2)$, for $u_1,u_2 \in [0,1]$, is the joint distribution function between the random variables $X_{(1)}$ and $X_{(2)}$ where both marginals are uniformly distributed. The bivariate copula function is given by,
\begin{equation}
C(u_1,u_2)=F_{X_{(1)},X_{(2)}}\left(F^{-1}_{X_{(1)}}(u_1),F^{-1}_{X_{(2)}}(u_2)\right),
\label{equation:copula}
\end{equation}
where $F_{X_{(1)},X_{(2)}}(x_{(1)},x_{(2)})=P(X_{(1)} \leq x_{(1)}, X_{(2)} \leq x_{(2)})$ is the joint cumulative distribution function (CDF) of $X_{(1)}$ and $X_{(2)}$, and $F^{-1}_{X_{(1)}}(u_{1})$ and $F^{-1}_{X_{(2)}}(u_2)$ are the marginal generalized inverse CDFs (quantile functions). They are defined by, $$
\inf_{x_{(1)}  \in \mathbb{R}}F_{X_{(1)}}(x_{(1)}) \geq u_1 \quad \text{and} \quad \inf_{x_{(2)}\in\mathbb{R}}F_{X_{(2)}}(x_{(2)}) \geq u_2,
$$
respectively \citep{Charpentier}. Let $\big\{x_{(1)}^i\big\}_{i=1}^{n}$ and $\big\{x_{(2)}^i\big\}_{i=1}^{n}$ be realisations (the data) of the random variables $X_{(1)}$ and $X_{(2)}$ respectively. In this case, a nonparametric empirical copula function is typically found to represent the dependence between the two data-sets in $\big\{x_{(1)}^i,x_{(2)}^i\big\}_{i=1}^{n}$. The empirical copula \citep{Deheuvels} converges to the true copula function in (\ref{equation:copula}), within the limit of $n \to \infty$. It is defined by,
\begin{equation}
\hat{C}(u_1,u_2)=\frac{1}{n}\sum^{n}_{i=1}\left(1_{x_{(1)}^i \leq \tilde{x}^{\ceil*{u_1 n}}_{(1)}}\right) \left( 1_{x_{(2)}^i \leq \tilde{x}_{(2)}^{\ceil*{u_2 n}}} \right).
\label{equation:indicatorcopula}
\end{equation}
where $\ceil*{z} = \text{arg}\min_{m \in \mathbb{Z}}( m \geq z)$ , $1$ is the indicator function and $\tilde{x}_{(k)}^i$ is the $i$'th order statistic of $\big\{x_{(k)}^i\big\}_{i=1}^{n}$, for $k=1,2$. Here, the mean over the products of indicator functions represents the approximation to the joint CDF of $X_{(1)}$ and $X_{(2)}$ in (\ref{equation:copula}), whilst the ranked order statistics $\tilde{x}_{(k)}^{\ceil*{u_k n}}$ within each indicator function represents the approximation to the marginal generalized inverse CDFs in (\ref{equation:copula}). Since each term in the summation in (\ref{equation:indicatorcopula}) will only be non-zero if $x_{(k)}^i \leq \tilde{x}_{(k)}^{\ceil*{u_k n}}$, for both $k=1,2$, then it suffices to sum the indicator functions corresponding to only one value of $k$ over all indices that satisfy $x_{(l)}^i \leq \tilde{x}_{(l)}^{\ceil*{u_l n}}$ for $l=\{1, 2\}\setminus k$. By doing this, the empirical copula can alternatively be written as \citep{Gregory},
\begin{equation}
\hat{C}(u_1,u_2)=\frac{n_1}{n}\sum_{i \in I}1_{x_{(2)}^i \leq \tilde{x}_{(2)}^{\ceil*{u_2 n}}},
\label{equation:newformcopulaind}
\end{equation}
where $n_1$ is the cardinality of the set $I \subset [1,n]$, such that $x_{(1)}^{I}$ correspond to elements in $\big\{x_{(1)}^{i}\big\}_{i=1}^{n}$ that satisfy $x_{(1)}^{I} \leq \tilde{x}_{(1)}^{\ceil*{u_1 n}}$. Note that $u_1$ is now accounted for via the set $I$. Now let $\hat{F}_{n,(2)}^{-1}(u_2)$ be the empirical quantile function \citep{Ma} of $\big\{x_{(2)}^{i}\big\}_{i=1}^{n}$, for which we will use the approximation,
\begin{equation}
\hat{F}_{n,(2)}^{-1}(u_2)=\tilde{x}_{(2)}^{\ceil*{u_2 n}}
\label{equation:quantilefunctionapprox},
\end{equation}
to the quantile function $F^{-1}_{X_{(2)}}(u_2)$. Also let $\hat{F}_{n_1,(2)}(x)$ be the empirical CDF given by, $$\hat{F}_{n_1,(2)}(x):=\frac{1}{n_1}\sum^{n_1}_{i=1}1_{x_{(1)}^{I(i)} \leq x},$$ where $I(i)$ is the $i$'th element of $I$\footnote{Note that this expression is only taken over elements in the set $I$}. Using this notation, another way of writing (\ref{equation:newformcopulaind}) is,
\begin{equation}
\hat{C}(u_1,u_2)=\frac{n_1}{n}\hat{F}_{n_1,(2)}\left(\hat{F}^{-1}_{n,(2)}(u_2)\right).
\label{equation:newformcopula}
\end{equation}
For more information on the statistical explanation behind the approximation in (\ref{equation:newformcopula}), turn to \cite{Gregory}.

One of the by-products of a copula is the computation of the tail dependence coefficients (there is an upper and a lower one) \citep{Sibuya}. These coefficients allow one to study the dependence in the tails of each of the marginals $X_{(1)}$ and $X_{(2)}$. For example, $X_{(1)}$ and $X_{(2)}$ may have low dependence over their entire probability space, however they could have very high dependence when both $X_{(1)}$ and $X_{(2)}$ take extreme values. This aspect of dependence is important in many applications, for example in financial analysis where it is crucial to realise if two assets have a high relative probability of both crashing at similar times.
The lower tail dependence coefficient between $X_{(1)}$ and $X_{(2)}$ can be computed directly via the copula function,
\begin{equation}
\lambda_L = \lim_{u \to 0}\frac{C(u,u)}{u},
\label{equation:ltdc}
\end{equation}
and so too the upper tail dependence coefficient,
\begin{equation}
\lambda_U = \lim_{u \to 1}\frac{1 -2u + C(u,u)}{(1-u)}.
\label{equation:utdc}
\end{equation}
The focus of this paper is on the estimation of empirical tail dependence coefficients for streaming data; this is done by utilising nonparametric empirical copulas. Under true model assumptions, one can estimate tail dependence by fitting a suitable parametric copula $C(u_1, u_2)$ to the data and this could offer a better estimate of the tail dependence coefficients than using any nonparametric copula function. An example of such a parametric copula is the Gaussian copula, where Gaussianity of the data is assumed. However in practice, if the assumptions on the data stream are too restrictive, then it is suitable to consider a nonparametric estimate of the underlying copula function. Estimators of the tail dependence coefficients that utilise these nonparametric copula estimates $\hat{C}(u_1,u_2)$, that fall into the scope of this paper, can be used in this case.  For a detailed account of these estimates, see \cite{Frahm} and \cite{Schmidt}. One could employ a model goodness-of-fit technique (e.g. Chi-Square), such as a test for Gaussianity in the case of considering a Gaussian copula, to confirm if this is the case. If it is more suitable to employ a parametric copula, the computational challenges associated with updating the tunable parameter estimates with indefinite data streams would have to be considered by the user, e.g. see \cite{Gregory} for more detail on this. On the other hand if a nonparametric copula estimate is a reasonable choice in a particular context, then the proposed methodology in this paper can be used to provide approximations of tail dependence coefficients over a data stream. This is the assumption that we make for the remainder of the paper.

Using the empirical copula $\hat{C}(u_1,u_2)$, one estimate of the empirical lower tail dependence coefficient is given by \citep{Caillault},
\begin{equation}
\hat{\lambda}_{L} = \lim_{i \to 0}\frac{\hat{C}(i/n,i/n)}{i/n},
\label{equation:lowerlimit}
\end{equation}
and one estimate of the empirical upper tail dependence coefficient is given by,
\begin{equation}
\hat{\lambda}_{U} = \lim_{i \to n}\frac{1-\frac{2i}{n}+\hat{C}(i/n,i/n)}{1-\frac{i}{n}} .
\label{equation:upperlimit}
\end{equation}
These are consistent with the tail dependence coefficients in (\ref{equation:ltdc}) and (\ref{equation:utdc}) respectively as $n \to \infty$, since the empirical copula is also consistent \citep{Deheuvels}. Empirically one cannot take this limit and therefore it suffices to study the following functions,
\begin{equation}
\hat{\lambda}_{L}(i/n)=\frac{\hat{C}(i/n,i/n)}{i/n},
\label{equation:dependencefunction}
\end{equation}
and
\begin{equation}
\hat{\lambda}_{U}(i/n)=\frac{1-\frac{2i}{n}+\hat{C}(i/n,i/n)}{1-\frac{i}{n}}.
\label{equation:upperdependencefunction}
\end{equation}
For $i=n-1,$ $n-2,\ldots,$ $2,$ $1$, the functions in (\ref{equation:dependencefunction}) and (\ref{equation:upperdependencefunction}) describe the path of $\hat{\lambda}_L$ and $\hat{\lambda}_U$ as $i$ tend to 1 and $n$ respectively \citep{Caillault}. Note that when evaluating these functions with a fixed value of $i$, they are consistent with the limits in (\ref{equation:lowerlimit}) and (\ref{equation:upperlimit}) as $n \to \infty$, since $\lim_{n \to \infty}(i/n)=0$. It has been proposed to evaluate the functions in (\ref{equation:dependencefunction}) and (\ref{equation:upperdependencefunction}) with the minimum and maximum values of $i$ that the functions are decreasing and increasing for respectively \citep{Caillault}. However for the scope of this paper, which will estimate these functions for an arbitrary fixed value of $i$, this particular selection is not justified further. This paper will instead concentrate on the approximation of the functions in (\ref{equation:dependencefunction}) and (\ref{equation:upperdependencefunction}) when the empirical copula function $\hat{C}(u_1,u_2)$ must be constructed over a data stream. The next section will propose an approximation to the tail dependence functions through approximation to the empirical copula, formed via a succinct summary of the data stream.

\subsection{Bivariate copula summaries for streaming data}

\label{sec:streamingdata}

Streaming data is the scenario in which say, the bivariate data stream $\big\{x_{(1)}^i,x_{(2)}^i\big\}_{i=1}^{n}$, is added to sequentially over (possibly indefinite) time. In the context of streaming data, it is not possible to store all of the data points in the stream or be able to consistently re-compute the order statistics for the empirical copula function (as mentioned above). This is typically due to restrictions on runtime and/or memory/storage. Quantile summaries are a common way of maintaining an approximation to the empirical quantile function in (\ref{equation:quantilefunctionapprox}) as an univariate data stream is added to, whilst only storing a succinct number of elements from the stream in space-memory \citep{Greenwald}. On this note, define an $\epsilon$-approximate quantile summary $Q$, to be an approximation to the quantile function $\hat{F}_{n,(k)}^{-1}(u)=\tilde{x}_{(k)}^{\ceil*{un}}$, that returns a value $\tilde{x}_{(k)}^j$, for $k=1,2$, where $j \in [\ceil*{un}-\epsilon n, \ceil*{un}+\epsilon n]$. An algorithm to construct such an summary was given in \cite{Greenwald}.

The work in \cite{Gregory} proposed another summary made up of a series of different $\epsilon$-approximate quantile summaries. This \textit{copula summary} maintained an approximation $\tilde{C}(u_1,u_2)$ to the bivariate empirical copula function in (\ref{equation:newformcopula}) over the data stream $\big\{x_{(1)}^i,x_{(2)}^i\big\}_{i=1}^{n}$. It was shown that an approximation within $[\hat{C}(u_1,u_2)-5 \epsilon,\hat{C}(u_1,u_2)+5\epsilon]$ can be achieved. Just like the univariate quantile summaries that it is composed from, the copula summary was space-efficient and stored only a succinct number of elements from the data stream. The extension to this summary, proposed in Sec. \ref{sec:modificationsummary}, will allow an approximation to empirical tail dependence coefficients of data streams to be computed.

\section{Estimation of tail dependence for data streams}

\label{sec:estimationtaildependence}


The copula summary presented in \cite{Gregory} was not suitable to find coefficients of tail dependence for one main reason: the error of the approximation to (\ref{equation:newformcopula}) was uniform over a grid of evaluation points $u_1$ and $u_2$. Therefore the resolution of the approximation would be as refined in the tails of the two marginals as it would be for the medians of both marginals. This results in an approximation to the tail dependence coefficients (replacing $\hat{C}(i/n,i/n)$ with $\tilde{C}(i/n,i/n)$ respectively in (\ref{equation:dependencefunction}) and (\ref{equation:upperdependencefunction})) that has error growing linearly with $n$. One can see this from the following error bound of the lower tail dependence coefficient function for fixed $i$,
\begin{equation}
\abs*{\hat{\lambda}_L(i/n)-\frac{\tilde{C}(i/n,i/n)}{i/n}} = \frac{n}{i}\abs*{\hat{C}(i/n,i/n)-\tilde{C}(i/n,i/n)}  \leq \frac{5\epsilon n}{i},
\label{equation:growthoferror}
\end{equation}
and the error bound of the upper tail dependence function for fixed $j=n-i$,
\begin{equation}
\abs*{\hat{\lambda}_U(i/n)-\frac{1-\frac{2i}{n}+\tilde{C}(i/n,i/n)}{1-(i/n)}} = \frac{n}{j}\abs*{\hat{C}(i/n,i/n)-\tilde{C}(i/n,i/n)}  \leq \frac{5\epsilon n}{j}.
\label{equation:uppergrowthoferror}
\end{equation}
Simply refining the prescribed error $\epsilon$ would be insufficient; one would need to sequentially refine $\epsilon$ as the stream gets longer, tending towards 0. The work presented in this paper is inspired by \cite{Cormode} in the univariate context, which considered biased quantile estimation and modified $\epsilon$-approximate quantile summaries to refine the error sufficiently at the tails, at the expense of error not at the tails. This will, as is apparent from the error analysis later in the paper (see Sec. \ref{sec:theoreticalanalysis}), guarantee that the error from the approximations of the tail dependence coefficients stay fixed as the number of elements in the stream is increased.

\subsection{Modifications to the copula summary}

\label{sec:modificationsummary}

This section details the specific modifications to the $\epsilon$-approximate quantile summary introduced in Sec. \ref{sec:streamingdata}, and therefore the copula summary, in order to obtain a suitable approximation to tail dependence coefficients. The proposed algorithm in \cite{Cormode} constructed a summary to maintain an approximation, with guaranteed error bounds, to the `biased' quantiles $u^l$, for $l=1,$ 2, $\ldots$, and $u \in [0,1]$. An approximation to the biased quantiles should have error relative to the quantile query, such that the approximation to $u^l$ should have an error of $\pm \epsilon u^l$ (or if required, $\pm \epsilon (1-u^l)$ by symmetry) rather than the uniform error of $\pm \epsilon$ from the standard quantile summary. This relative error allows one to refine the quantile approximation within the tails of an univariate distribution, and therefore is suited to the problem considered in this paper. We modify the construction of the summary proposed in \cite{Cormode} slightly to allow an error of $\pm \epsilon \min(u,1-u)$, assuming $l=1$, given that we would like to approximate both tails with relative error. On this note, define a $\epsilon u$-approximate quantile summary to be an approximation to $\hat{F}_{n,(k)}^{-1}(u)=\tilde{x}^{\ceil{un}}_{(k)}$, for $k=1,2$, which returns the value $\tilde{x}^{j}$, where \begin{equation}j \in [\ceil{un}-\epsilon \min(u, (1-u))n, \ceil{un}+\epsilon \min(u, (1-u)) n].\label{equation:conditionforuapprox}\end{equation}

Recall that the copula summary is composed of $L+1$ quantile summaries: $\big\{S_{(1)},$ $S_{(2)}^1,\ldots,$ $S_{(2)}^{L}\big\}$. Proved later in Sec. \ref{sec:theoreticalanalysis}, it will suffice to let the summaries $\big\{S_{(1)},$ $S_{(2)}^1,\ldots,$ $S_{(2)}^{L}\big\}$ be modified into $\epsilon u$-approximate summaries in order to obtain suitable approximations to the tail dependence coefficients. To modify them, the manner in which each summary is maintained and queried (through Sec. \ref{sec:inserting}, \ref{sec:combining} and \ref{sec:querying}) is changed from that explained in \cite{Gregory}. First, the basic make-up of each quantile summary $Q$ remains the same from the initial summary proposed back in \cite{Greenwald}. Each summary is made up of tuples, $Q=\big\{(v_i,g^i,\Delta^i)\big\}_{i=1}^{L}$. The values $v_i$, where $v_{i-1} \leq v_{i} < v_{i+1}$, are elements that have been seen in the data stream $\big\{x_{(k)}^i\big\}_{i=1}^n$, for $k=1,2$, so far. These values are maintained by the summary as `cover' for a range of quantiles that one may wish to query. I.e. the value $v_i$ will be returned as an approximation to a nearby quantile. The values $g^i$ and $\Delta^i$ control the range of quantiles that the value $v_i$ is returned as an approximation to. They do this by governing the minimum, $r_{min,Q}(v_i)$, and maximum, $r_{max,Q}(v_i)$ rank that the value $v_i$ takes in the original data stream. We define
\begin{equation}
g^i=r_{min,Q}(v_i)-r_{min,Q}(v_{i-1}),
\label{equation:gis}
\end{equation}
and then,
$$
r_{min,Q}(v_i)=\sum^i_{j=1}g^j, \qquad r_{max,Q}(v_i)=r_{min,Q}(v_i)+\Delta^i.
$$
One also knows the length of the data stream at any one time via $n=\sum^L_{j=1} g^j$. In order to guarantee that the $\epsilon u$-approximate quantile summary maintains an approximation which satisfies (\ref{equation:conditionforuapprox}), these minimum and maximum ranks must satisfy \citep{Cormode},
\begin{equation}
r_{max,Q}(v_{i+1})-r_{min,Q}(v_i) \leq 2\epsilon \min\left( r_{min,Q}(v_i), n-r_{max,Q}(v_{i+1})\right).
\label{equation:empiricalrankcondition}
\end{equation}
See Appendix \ref{sec:proofofnewinvariant} for a proof of this.

In the copula summary, each of the \textit{subsummaries} $S_{(2)}^1$, $S_{(2)}^2,\ldots$, $S_{(2)}^L$ corresponds to an element inside of $S_{(1)}$ (therefore $S_{(1)}$ has the cardinality $L$). Whilst the summary $S_{(1)}$ contains elements (and information about their ranks) from the first component of the data stream, i.e. $\big\{x_{(1)}^i\big\}_{i=1}^n$, the summaries $S_{(2)}^1$, $S_{(2)}^2,\ldots$, $S_{(2)}^L$ contain elements (and information about their ranks) from the second component of the data stream, i.e. $\big\{x_{(2)}^{i}\big\}_{i=1}^{n}$. On this note, let a tuple in $S_{(1)}$ be denoted by $(v_i,g_{(1)}^i,\Delta_{(1)}^i)$ where $v_i \in \big\{x_{(1)}^i\big\}_{i=1}^{n}$, for $i=1,\ldots,L$, and let a tuple in $S_{(2)}^i$ be denoted by $(w_i,g_{(2)}^{i,j},\Delta_{(2)}^{i,j})$ where $w_i \in \big\{x_{(2)}^{i}\big\}_{i=1}^{n}$, for $j=1,\ldots,L_{i}$ (therefore $S_{(2)}^i$ has the cardinality $L_i$). The parameters within these summaries are changed carefully over time as new elements are added to the data stream $\big\{x_{(1)}^i,x_{(2)}^i\big\}_{i=1}^{n}$ via the operations defined in the following three sections. It is important to remember that the proposed methodology here can only be used with data streams where $x_{(1)}^i$ and $x_{(2)}^i$ are acquired at the same frequency and the same time. This is because the operations described in the following sections take a pair of data points as an input. Considering the case where the data streams components are acquired at different rates is outside the scope of this paper; however this could be achieved by modifying the times at which the summary $S_{(1)}$ and the subsummaries $S_{(2)}^1,S_{(2)}^2,\ldots,S_{(2)}^L$ get updated using the operations described below.

\subsection{Inserting an element into the copula summary}
\label{sec:inserting}

When the element $(x_{(1)}^{n+1},x_{(2)}^{n+1})$ enters the data stream, it should be inserted into the copula summary. This is done by inserting the tuple $(x_{(1)}^{n+1},1,\Delta_{(1)}^*)$ into $S_{(1)}$. If $x_{(1)}^{n+1} < v_1$, then we insert the tuple at the start of $S_{(1)}$, and let $\Delta_{(1)}^*=1$. Conversely if $x_{(1)}^{n+1} \geq v_L$, then we insert the tuple at the end of $S_{(1)}$ and let $\Delta_{(1)}^*=1$. If $v_i \leq x_{(1)}^{n+1} <v_{i+1}$, then insert the tuple in between $v_i$ and $v_{i+1}$ and let $\Delta_{(1)}^*=\min(g_{(1)}^i+\Delta_{(1)}^i-1, g_{(1)}^{i+1}+\Delta_{(1)}^{i+1}-1)$. Now for the second component of the new element, let $S_{(2)}^{*}=\big\{(x_{(2)}^{n+1},1,0)\big\}$ be a new quantile summary.  This summary gets inserted between $S_{(2)}^{i}$ and $S_{(2)}^{i+1}$ in the copula summary $\big\{S_{(1)},S_{(2)}^{1},\ldots,S_{(2)}^{L}\big\}$ if $v_i \leq x_{(1)}^{n+1} < v_{i+1}$, between $S_{(1)}$ and $S_{(2)}^{1}$ if $x_{(1)}^{n+1}<v_1$ or at the end of the copula summary if $x_{(1)}^{n+1} \geq v_{L}$. Finally, increase $L$ by 1.

\subsection{Combining tuples in the copula summary}
\label{sec:combining}

Combining the tuples in the summaries $\big\{S_{(1)},S_{(2)}^1,\ldots,S_{(2)}^{L}\big\}$ is occasionally required to remove unnecessary tuples from the summaries, whilst maintaining the elements required for the approximation to be of the desired accuracy. Providing that $L > 3$, sequentially for each element $i \in [3, L - 1]$ in $S_{(1)}$ we find the index $j \in [2,i]$ satisfying
\begin{equation}
\text{arg} \min_{j}\Bigg\{\sum^{i}_{k=j}g_{(1)}^k+\Delta_{(1)}^i \leq 2\epsilon\min\left( r_{min,S_{(1)}}(v_{j-1}),n-r_{max,S_{(1)}}(v_{i})\right) =  2 \epsilon \min\left(\sum^{j-1}_{k=1}g_{(1)}^{k},n-\sum^{i}_{k=1}g_{(1)}^{k}-\Delta_{(1)}^{i}\right)\Bigg\}.
\label{equation:combinationcondition}
\end{equation}
Once this value is found, the tuples $(v_j, g_{(1)}^j, \Delta_{(1)}^j), \ldots,$ $(v_i,g_{(1)}^i,\Delta_{(1)}^i)$ can be combined into the new tuple $(v_i,\sum^i_{k=j}g_{(1)}^k,\Delta_{(1)}^i)$. We use the condition on $j$ in (\ref{equation:combinationcondition}) in order to guarantee (\ref{equation:empiricalrankcondition}) is satisfied. In addition to combining those tuples, we \textit{merge} the tuples $S_{(2)}^j,\ldots,$ $S_{(2)}^i$ into a new tuple $M(S_{(2)}^j,\ldots,S_{(2)}^i)$. See Sec. \ref{sec:mergingproof} for the implementation of this merging and the bound of the approximation error. Finally combine unnecesary tuples (e.g. $(w_\cdot,g_{(2)}^{i,\cdot},\Delta_{(2)}^{i,\cdot})$) inside of this merged summary $M(S_{(2)}^j,\ldots,S_{(2)}^i)$, in the manner described earlier in this section. Insert this new summary in the place of $S_{(2)}^j,\ldots,$ $S_{(2)}^i$ in the copula summary, such that the copula summary now is $\big\{S_{(1)},S_{(2)}^1,\ldots,S_{(2)}^{j-1}, M(S_{(2)}^j,\ldots,S_{(2)}^i), S_{(2)}^{i+1},\ldots\big\}$.

\subsection{Querying the copula summary}
\label{sec:querying}

This section now describes how to query the copula summary (maintained over time using the operations in Sec. \ref{sec:inserting} and \ref{sec:combining}) for an approximation to $\hat{C}(u_1,u_2)$. We will denote this approximation by $\tilde{\tilde{C}}(u_1,u_2)$, as opposed to the approximation from the copula summary proposed in \cite{Gregory} composed of $\epsilon$-approximate quantile summaries, $\tilde{C}(u_1,u_2)$. First, we compute the approximation to the empirical quantile function $\hat{F}^{-1}_{n,(1)}(u_1)=\tilde{x}_{(1)}^{\ceil*{u_1 n}}$ using the $\epsilon u$-approximate quantile summary $S_{(1)}$; denote this approximation by $\tilde{\tilde{F}}^{-1}_{n,(1)}(u_1)$. Let $E$ be equal to the value of $i$ that satisfies $\tilde{\tilde{F}}_{n,(1)}^{-1}(u_1)=v_i$, and find the total number of elements in the stream that have entered into the first $E$ subsummaries $S_{(2)}^1,\ldots,$ $S_{(2)}^E$,
\begin{equation}
\hat{n}_1=\sum^E_{i=1}g_{(1)}^{i}=\sum^{E}_{i=1}\sum^{L_i}_{j=1}g_{(2)}^{i,j}.
\label{equation:nhat}
\end{equation}
Suppose that the indices of the $\hat{n}_1$ elements to have entered into the first $E$ subsummaries form the set $\hat{I} \subset [1,n]$; note this is an approximation to the set $I$ introduced in Sec. \ref{sec:empiricalcopulas}.

Next, let $M(S_{(2)}^1,\ldots,S_{(2)}^L)$ be a merged summary composed of all the subsummaries $S_{(2)}^1,\ldots,$ $S_{(2)}^L$ (again for the implementation details of this merge, see Sec. \ref{sec:mergingproof}). This summary can be queried for an approximation to $\hat{F}_{n,(2)}^{-1}(u_2)$; denote this approximation by $\tilde{\tilde{F}}_{n,(2)}^{-1}(u_2)$. Finally, let $M(S_{(2)}^1,\ldots,S_{(2)}^E)$ be a merged summary of the subsummaries $S_{(2)}^1,\ldots,$ $S_{(2)}^E$. Then define the approximation $\tilde{\tilde{F}}_{\hat{n}_1,(2)}(y)$ to the empirical CDF $$\hat{F}_{\hat{n}_1,(2)}(y)=\frac{1}{\hat{n}_1}\sum^{\hat{n}_1}_{i=1}1_{\tilde{x}^{\hat{I}(i)}_{(2)} \leq y},$$ to be an \textit{inverse} query of the summary $M(S_{(2)}^1,\ldots,S_{(2)}^E)$. The implementation details of this query, and a guarantee on it's error with respect to the empirical CDF, is provided later in Sec. \ref{sec:inversequeryapprox}. In total, the copula summary approximation $\tilde{\tilde{C}}(u_1,u_2)$, to the empirical copula function $\hat{C}(u_1,u_2)$ is given by,
\begin{equation}
\tilde{\tilde{C}}(u_1,u_2)=\frac{\hat{n}_1}{n}\tilde{\tilde{F}}_{\hat{n}_1,(2)}\left(\tilde{\tilde{F}}_{n,(2)}^{-1}(u_2)\right).
\label{equation:copulaapproxu}
\end{equation}


\section{Theoretical analysis of the approximation}

\label{sec:theoreticalanalysis}

This section provides a theoretical analysis of the approximation in (\ref{equation:copulaapproxu}) to the empirical copula function, and the resulting approximations to the empirical tail dependence coefficients. First, it is important to clarify error bounds for merged $\epsilon u$ - approximate quantile summaries, and an inverse query of a $\epsilon u$ - approximate quantile summary. Recall from (\ref{equation:empiricalrankcondition}) that the summary $Q$ is a $\epsilon u$-approximate quantile summary if two neighbouring elements $v_{i+1}$ and $v_i$ in $Q$ satisfy,
\begin{equation}
r_{max,Q}(v_{i+1})-r_{min,Q}(v_i) \leq 2 \epsilon \min\left( r_{min,Q}(v_i), n-r_{max,Q}(v_{i+1})\right).
\label{equation:newcondition}
\end{equation}

The next two sections cover two preliminary bounds, before Sec. \ref{sec:errorboundtail} outlines the error bound of the approximation to the tail dependence coefficients using the modified copula summary, proposed in this paper.

\subsection{Merging $\epsilon u$-approximate quantile summaries}

\label{sec:mergingproof}

We recall from \cite{GreenwaldMerge} that one can merge $\epsilon$-approximate quantile summaries $Q_1$ (length $L_1$) and $Q_2$ (length $L_2$) to obtain the quantile summary $M(Q_1,Q_2)$, containing the elements $Q_1 \bigcup Q_2$, which is also $\epsilon$-approximate itself. It does this via the following method. Suppose $z_k$, for $k \in [1, L_1 + L_2]$, is an element from $Q_1$ which exists in $M(Q_1,Q_2)$. If it exists, let $w_1$ be the largest element in $Q_2$ that is less than or equal to $z_k$. Also, if it exists, let $w_2$ be the smallest element in $Q_2$ that is greater than $z_k$. Then set,
\begin{equation}
r_{min,M(Q_1,Q_2)}(z_k)=
\begin{cases}
r_{min,Q_2}(w_1) + r_{min,Q_1}(z_k), & \text{if } w_1 \text{ exists}\\
r_{min,Q_1}(z_k), & \text{otherwise},
\end{cases}
\end{equation}
and
\begin{equation}
r_{max,M(Q_1,Q_2)}(z_k)=
\begin{cases}
r_{max,Q_2}(w_2) + r_{max,Q_1}(z_k) - 1, & \text{if } w_2 \text{ exists}\\
r_{max,Q_2}(w_1) + r_{max,Q_1}(z_k), & \text{otherwise}.
\end{cases}
\end{equation}
It will now be proved that if $Q_1$ and $Q_2$ are $\epsilon u$-approximate then $M(Q_1,Q_2)$ is a $\epsilon u$-approximate summary as well. It is therefore necessary to show that $$r_{max,M(Q_1,Q_2)}(z_{k+1})-r_{min,M(Q_1,Q_2)}(z_k)\leq 2 \epsilon \min\left(r_{min,M(Q_1,Q_2)}(z_k),n-r_{max,M(Q_1,Q_2)}(z_{k+1})\right).$$ Suppose $Q_1$ and $Q_2$ were constructed over data streams of length $n_1$ and $n_2$ respectively.
For the case where $z_k$ and $z_{k+1}$ are from the same summary, let them equal $x_1$ and $x_{2}$ (say in the summary $Q_1$). If there exists both the elements $w_1$ (largest element in $Q_2$ that is less than or equal to $z_k$) and $w_2$ (smallest element in $Q_2$ that is greater than $z_{k+1}$) then we know that $w_1$ and $w_2$ are consecutive elements in $Q_2$. Thus,
\begin{equation}
\begin{split}
r_{max,M(Q_1,Q_2)}(z_{k+1})-r_{min,M(Q_1,Q_2)}(z_k) &\leq (r_{max,Q_1}(x_{2})-r_{min,Q_1}(x_{1})) + (r_{max,Q_2}(w_2)-r_{min,Q_2}(w_{1}) - 1) \\
\quad &\leq 2 \epsilon \big(\min\left(r_{min,Q_1}(x_{1}),n_1-r_{max,Q_1}(x_{2})\right)+\ldots\\
\quad & \ldots \min\left(r_{min,Q_2}(w_1),n_2-r_{max,Q_2}(w_2)\right)\big)\\
\quad &\leq2 \epsilon \min\left(r_{min,M(Q_1,Q_2)}(z_{k}),(n_1+n_2)-r_{max,M(Q_1,Q_2)}(z_{k+1})\right).
\end{split}
\end{equation}

If $w_1$ doesn't exist, then
\begin{equation*}
\begin{split}
r_{max,M(Q_1,Q_2)}(z_{k+1})-r_{min,M(Q_1,Q_2)}(z_k) &\leq (r_{max,Q_1}(x_2) - r_{min,Q_1}(x_1)) \leq 2 \epsilon \min\left(r_{min,Q_1}(x_1),n_1-r_{max,Q_1}(x_2)\right) \\
\quad &\leq  2\epsilon \min \left(r_{min,M(Q_1,Q_2)}(z_k), (n_1+n_2)-r_{max,M(Q_1,Q_2)}(z_{k+1})\right),
\end{split}
\end{equation*}
as $r_{max,Q_2}(w_2)=1$ and if $w_2$ doesn't exist, then
\begin{equation*}
\begin{split}
r_{max,M(Q_1,Q_2)}(z_{k+1})-r_{min,M(Q_1,Q_2)}(z_k) &\leq (r_{max,Q_2}(w_1) + r_{max,Q_1}(x_2)) - (r_{min,Q_2}(w_1)) + r_{min,Q_1}(x_1)) \\
\quad &\leq 2 \epsilon \big(\min\left(r_{min,Q_1}(x_{1}), n_1-r_{max,Q_1}(x_2)\right)+\ldots\\
\quad &\min\left(r_{min,Q_2}(w_1),n_2-r_{max,Q_2}(w_1)\right)\big)\\
\quad &\leq 2 \epsilon \min\left(r_{min,M(Q_1,Q_2)}(z_{k}), (n_1+n_2)-r_{max,M(Q_1,Q_2)}(z_{k+1})\right).
\end{split}\end{equation*}
For the case where $z_k$ and $z_{k+1}$ come from different summaries, say $z_k$ from $Q_1$ labelled by $x_1$, and $z_{k+1}$ from $Q_2$ labelled by $w_2$. Then w.l.o.g. let $x_2$ be the smallest element in $Q_1$ greater than $w_2$, and $w_1$ be the largest element in $Q_2$ less than or equal to $x_1$. Then,
\begin{equation*}
\begin{split}
r_{max,M(Q_1,Q_2)}(z_{k+1})-r_{min,M(Q_1,Q_2)}(z_k) &\leq (r_{max,Q_1}(x_2)+r_{max,Q_2}(w_2)-1)-(r_{min,Q_1}(x_1)+r_{min,Q_2}(w_1))\\
\quad &\leq (r_{max,Q_1}(x_2)-r_{min,Q_1}(x_1)-1)-(r_{max,Q_2}(w_2)-r_{min,Q_2}(w_1)) \\
\quad & \leq 2 \epsilon \big(\min\left(r_{min,Q_1}(x_{1}),n_1-r_{max,Q_1}(x_2)\right)+\ldots\\
\quad &\min\left(r_{min,Q_2}(w_1),n_2-r_{max,Q_2}(w_2)\right)\big)\\
\quad &\leq 2 \epsilon \min\left(r_{min,M(Q_1,Q_2)}(z_{k}),(n_1+n_2)-r_{max,M(Q_1,Q_2)}(z_{k+1})\right). 
\end{split}\end{equation*}
We now have the condition for a $\epsilon u$-approximate summary in (\ref{equation:newcondition}) for all cases of membership to $Q_1$ and $Q_2$ for the elements $z_k$ and $z_{k+1}$. Therefore any $\epsilon u$-approximate summaries merged together will also be $\epsilon u$-approximate too.

\subsection{Inversely querying $\epsilon u$-approximate quantile summaries}

\label{sec:inversequeryapprox}

In this section, we would like to bound the approximation $\tilde{\tilde{F}}_{n,(k)}(x)$, for $x \in \mathbb{R}$ and $k = 1,2$, to the empirical CDF $\hat{F}_{n,(k)}(x)$ using a $\epsilon u$-approximate quantile summary $Q$ of the data stream $\big\{x^i_{(k)}\big\}_{i=1}^{n}$. This is a simple extension to the proof in \cite{Lall} for inversely querying a $\epsilon$-approximate quantile summary. Firstly, let $\hat{F}_{n,(k)}(x)=i/n$, meaning $\tilde{x}^{i}_{(k)} \leq x < \tilde{x}^{i+1}_{(k)}$, where $\tilde{x}^0_{(k)}=-\infty$ and $\tilde{x}^{n+1}_{(k)}=\infty$. Let $j\in [1,L]$, then we know that using the quantile summary we keep an approximation, $\tilde{x}^{i_j}_{(k)}$, to the $i$'th order statistic of the data stream; this approximation satisfies $\tilde{x}^{i_j}_{(k)} \leq \tilde{x}^{i}_{(k)} \leq x < \tilde{x}^{i+1}_{(k)} \leq \tilde{x}^{i_{j+1}}_{(k)}$. Note that as the summary is $\epsilon u$-approximate, we have $\abs{i_j - i} \leq 2 \epsilon \min(i, n-i)$. Also recall from Sec. \ref{sec:modificationsummary} that we can only access the minimum and maximum values that $i_j$ can take, and not actually $i_j$ itself. To find an approximation to $i_j$ we can simply search all values $v_l$ in $Q$, for $l=1, \ldots, L$, for the $l$ that satisfies $v_l \leq x < v_{l+1}$ (with $v_{L+1}=\infty$) and take $(r_{min,Q}(v_l)+r_{max,Q}(v_l))/2$ as $\hat{i}_j$ as the approximation to $i_j$. If $x < v_1$, of course take $\hat{i}_j=0$. Now we know that $r_{max,Q}(v_l)-r_{min,Q}(v_l) \leq r_{max,Q}(v_l) - r_{min,Q}(v_{l-1}) \leq  2 \epsilon \min(i,n-i)$, and therefore $\abs{\hat{i}_j-i_j} \leq \epsilon \min(i,n-i)$. Due to the triangle inequality we have that $\abs{\hat{i}_j-i} \leq 3 \epsilon \min(i,n-i)$ and finally that
\begin{equation}
\abs{\hat{F}_{n,(k)}(x)-\tilde{\tilde{F}}_{n,(k)}(x)} \leq 3 \epsilon \min(i,n-i) /n = 3 \epsilon \min\left(\hat{F}_{n,(k)}(x),1-\hat{F}_{n,(k)}(x)\right).
\label{equation:boundforinverse}
\end{equation}

\subsection{Bounding the error of the approximation to the tail dependence coefficients}

\label{sec:errorboundtail}



Now that we have covered some necessary bounds, we can derive the guaranteed error bound of the modified copula summary and therefore the tail dependence coefficients approximations. Recall that the $\epsilon u$-approximate empirical copula approximation is given by,
\begin{equation}
\tilde{\tilde{C}}(u_1,u_2) = \frac{\hat{n}_1}{n}\tilde{\tilde{F}}_{\hat{n}_1,(2)}(\tilde{\tilde{F}}^{-1}_{n,(2)}(u_2)),
\label{equation:newcopulaapprox}
\end{equation}
therefore the approximation to the lower tail dependence coefficient (for a fixed $i$) is given by,
\begin{equation}
\tilde{\lambda}_{L}(i/n) = \frac{\tilde{\tilde{C}}(i/n,i/n)}{i/n},
\label{equation:newapproximationtaildependence}
\end{equation}
and the approximation to the upper tail dependence coefficient (for a fixed $j=n-i$) is given by,
\begin{equation}
\tilde{\lambda}_{U}(i/n) = \frac{1-(2i/n)+\tilde{\tilde{C}}(i/n,i/n)}{1-(i/n)}.
\label{equation:newapproximationuppertaildependence}
\end{equation}
\begin{theorem}

Let $\tilde{\tilde{C}}(i/n,i/n)$ denote the $\epsilon u$-approximate copula summary approximation, given in (\ref{equation:newcopulaapprox}), to the bivariate empirical copula $\hat{C}(i/n,i/n)$, then one can bound the error of this approximation by,
\begin{equation}
\abs*{\tilde{\tilde{C}}(i/n,i/n)-\hat{C}(i/n,i/n)} \leq \frac{\epsilon \min(i,n-i)) (8 + 9\epsilon)}{n}.
\label{equation:modifiedcopulabound}
\end{equation}
Therefore the approximation to the lower and upper tail dependence functions in (\ref{equation:newapproximationtaildependence}) and (\ref{equation:newapproximationuppertaildependence}) can be bounded by,
\begin{equation}
\abs*{\tilde{\lambda}_{L}(i/n)-\hat{\lambda}_L(i/n)} \leq \epsilon (8 + 9\epsilon), \quad i \leq \ceil{n/2},
\label{equation:taildependencebound}
\end{equation}
and
\begin{equation}
\abs*{\tilde{\lambda}_{U}(i/n)-\hat{\lambda}_U(i/n)} \leq \epsilon (8 + 9\epsilon), \quad i>\ceil{n/2},
\label{equation:uppertaildependencebound}
\end{equation}
respectively. The approximation error is therefore constant with increasing $n$.
\end{theorem}
To prove this bound, we shall follow the steps of the proof in \cite{Gregory} with some modifications.
\begin{proof}
We shall split the error
$$
\abs*{\tilde{\tilde{C}}(i/n,i/n)-\hat{C}(i/n,i/n)}=\abs*{\frac{\hat{n}_1}{n}\tilde{\tilde{F}}_{\hat{n}_1,(2)}(\tilde{\tilde{F}}^{-1}_{n,(2)}(i/n))-\frac{n_1}{n}\hat{F}_{n_1,(2)}(\hat{F}^{-1}_{n,(2)}(i/n))},
$$
into three contributing parts via the triangle inequality and prove each individually,
\begin{equation}
\begin{split}
\abs*{\frac{\hat{n}_1}{n}\tilde{\tilde{F}}_{\hat{n}_1,(2)}(\tilde{\tilde{F}}^{-1}_{n,(2)}(i/n))-\frac{n_1}{n}\hat{F}_{n_1,(2)}(\hat{F}^{-1}_{n,(2)}(i/n))} \leq & \underbrace{\abs*{\frac{\hat{n}_1}{n}\tilde{\tilde{F}}_{\hat{n}_1,(2)}(\tilde{\tilde{F}}^{-1}_{n,(2)}(i/n))-\frac{\hat{n}_1}{n}\hat{F}_{\hat{n}_1,(2)}(\tilde{\tilde{F}}^{-1}_{n,(2)}(i/n))}}_{\text{A}}+\\
\quad &  \underbrace{\abs*{\frac{\hat{n}_1}{n}\hat{F}_{\hat{n}_1,(2)}(\tilde{\tilde{F}}^{-1}_{n,(2)}(i/n))-\frac{\hat{n}_1}{n}\hat{F}_{\hat{n}_1,(2)}(\hat{F}^{-1}_{n,(2)}(i/n))}}_{\text{B}}+ \\
\quad & \underbrace{\abs*{\frac{\hat{n}_1}{n}\hat{F}_{\hat{n}_1,(2)}(\hat{F}^{-1}_{n,(2)}(i/n))-\frac{n_1}{n}\hat{F}_{n_1,(2)}(\hat{F}^{-1}_{n,(2)}(i/n))}}_{\text{C}}.
\end{split}
\label{equation:contributing}
\end{equation}
\end{proof}
The first contributing part in (\ref{equation:contributing}) corresponds to the error associated with using an inverse query of the $\epsilon u$-approximate quantile summary $M(S_{(2)}^1,\ldots,S_{(2)}^E)$, instead of the empirical CDF $\hat{F}_{\hat{n}_1, (2)}$.
\begin{theorem}
The term (A) in (\ref{equation:contributing}) can be bounded by,
$$
\abs*{\frac{\hat{n}_1}{n}\tilde{\tilde{F}}_{\hat{n}_1,(2)}(\tilde{\tilde{F}}^{-1}_{n,(2)}(i/n))-\frac{\hat{n}_1}{n}\hat{F}_{\hat{n}_1,(2)}(\tilde{\tilde{F}}^{-1}_{n,(2)}(i/n))} \leq \frac{\epsilon \min(i,n-i) (6 + 9\epsilon)}{n}.
$$
\end{theorem}

\begin{proof}
From (\ref{equation:boundforinverse}), we can bound
\begin{equation}
\abs*{\tilde{\tilde{F}}_{\hat{n}_1,(2)}(\tilde{\tilde{F}}^{-1}_{n,(2)}(i/n))-\hat{F}_{\hat{n}_1,(2)}(\tilde{\tilde{F}}^{-1}_{n,(2)}(i/n))} \leq 3 \epsilon \min\left(\hat{F}_{\hat{n}_1,(2)}(\tilde{\tilde{F}}_{n,(2)}^{-1}(i/n)), 1-\hat{F}_{\hat{n}_1,(2)}(\tilde{\tilde{F}}_{n,(2)}^{-1}(i/n))\right).
\label{equation:theorem2bound}
\end{equation}
Let
$$
D = \hat{F}_{\hat{n}_1,(2)}(\tilde{\tilde{F}}_{n,(2)}^{-1}(i/n)).
$$
First we will concentrate on bounding this term by above and below, since (\ref{equation:theorem2bound}) is proportional to the minimum of $D$ and $1-D$. As $D$ depends on the query $\tilde{\tilde{F}}_{n,(2)}^{-1}(i/n)$, this query also needs to be bounded above and below.
\subsubsection*{Bounding $D$ by above}
Note that the query $\tilde{\tilde{F}}_{n,(2)}^{-1}(i/n)$ is an increasing function with $i$, and clearly $\hat{F}_{\hat{n}_1,(2)}(\tilde{\tilde{F}}_{n,(2)}^{-1}(i/n))$ is increasing with $\tilde{\tilde{F}}_{n,(2)}^{-1}(i/n)$ and $i$ too. Therefore it suffices to bound $D$ with $\hat{F}_{\hat{n}_1,(2)}$ evaluated at the upper bound for $\tilde{\tilde{F}}_{n,(2)}^{-1}(i/n)$. Note that as $M(S_{(2)}^1,\ldots,S_{(2)}^{L})$ is a $\epsilon u$-approximate summary as shown in Sec. \ref{sec:mergingproof}, we have that,
$$
\tilde{\tilde{F}}_{n,(2)}^{-1}(i/n) = \tilde{x}_{(2)}^{j}, \qquad j \in [i-\epsilon \min(i,n-i), i + \epsilon \min(i,n-i)],
$$
and therefore,
$$
\tilde{\tilde{F}}_{n,(2)}^{-1}(i/n) \leq \tilde{x}_{(2)}^{i + \floor{\epsilon \min(i,n-i)}}=G.
$$
Recall that $\hat{n}_1\hat{F}_{\hat{n}_1,(2)}(G)$ is just the count of all elements in $\big\{x_{(2)}^{\hat{I}(i)}\big\}_{i=1}^{\hat{n}_1}$ less than or equal to $\tilde{x}_{(2)}^{i+\floor{\epsilon \min(i, n-i)}}$.
Since $\big\{x_{(2)}^{\hat{I}(i)}\big\}_{i=1}^{\hat{n}_1} \subset \big\{x_{(2)}^{i}\big\}_{i=1}^{n}$, it must follow that $\hat{n}_1\hat{F}_{\hat{n}_1,(2)}(G) \leq n\hat{F}_{n,(2)}(G)\leq i + \floor{\epsilon \min(i, n-i)}$. Therefore,
$$
D \leq \frac{i + \min(i,n-i)\epsilon}{\hat{n}_1}.
$$
\subsubsection*{Bounding $D$ by below}
For a lower bound on $D$, we shall use the following bound\footnote{This is given by the fact that in the case where there is a perfect negative rank correlation between the two components in the data stream $\big\{x_{(1)}^i,x_{(2)}^i\big\}_{i=1}^{n}$, there will always be $\max(2i-n,0)$ values of $j$ that satisfy $x_{(1)}^j \leq \tilde{x}_{(1)}^{i}$ and $x_{(2)}^j \leq \tilde{x}_{(2)}^i$},
$$
n_1\hat{F}_{n_1,(2)}\left(\hat{F}_{n,(2)}(i/n)\right) \geq \max(2i-n, 0).
$$
 Now, recall from Sec. \ref{sec:querying} that $\hat{I} \subset I$ if $\hat{n}_1 < n_1$ or $I \subset \hat{I}$ if $\hat{n}_1 > n_1$; given that $S_{(1)}$ is a $\epsilon u$-approximate summary we have that,
$$
\hat{n}_1\hat{F}_{\hat{n}_1,(2)}\left(\hat{F}_{n,(2)}(i/n)\right) \geq \max\left(2i-n-\min(i,n-i)\epsilon, 0\right).
$$
In addition to this, note that as $M(S_{(2)}^1,\ldots,S_{(2)}^{L})$ is a $\epsilon u$-approximate summary, we have that $\tilde{\tilde{F}}_{n,(2)}^{-1}(i/n) \geq \tilde{x}_{(2)}^{i - \ceil{\epsilon \min(i,n-i)}}$ and,
$$
\hat{n}_1\hat{F}_{\hat{n}_1,(2)}\left(\tilde{\tilde{F}}_{n,(2)}(i/n)\right) \geq \max\left(2\left(i-\min(i,n-i)\epsilon\right)-n-\min(i,n-i)\epsilon, 0\right).
$$
Then finally,
$$
D \geq \max\left(\frac{2i-n-3\min(i,n-i)\epsilon}{\hat{n}_1}, 0\right),
$$
and therefore,
\begin{equation*}
\begin{split}
1 - D &\leq \min\left(\frac{2\left(n-i\right)+3\min(i,n-i)\epsilon}{\hat{n}_1}, 1 \right)\\
\quad &\leq \left(\frac{2\left(n-i\right)+3\min(i,n-i)\epsilon}{\hat{n}_1}\right).
\end{split}
\end{equation*}
\bigskip
Now that $D$ and $1-D$ have been bounded from above, these bounds can be used to bound the overall error in (\ref{equation:theorem2bound}).
\subsubsection*{Overall bound}
The term (A) in (\ref{equation:contributing}) can be bounded by,
\begin{equation*}
\begin{split}
\abs*{\frac{\hat{n}_1}{n}\tilde{\tilde{F}}_{\hat{n}_1,(2)}(\tilde{\tilde{F}}^{-1}_{n,(2)}(i/n))-\frac{\hat{n}_1}{n}\hat{F}_{\hat{n}_1,(2)}(\tilde{\tilde{F}}^{-1}_{n,(2)}(i/n))} &\leq 3\epsilon \min\left(\left(\frac{2(n-i)+3\min(i,n-i)\epsilon}{n}\right),\left(\frac{i+\min(i,n-i)\epsilon}{n}\right)\right)\\
\quad &\leq 3\epsilon\min\left(\left(\frac{2(n-i)+3\min(i,n-i)\epsilon}{n}\right), \left(\frac{2i+3\min(i,n-i)\epsilon}{n}\right)\right)\\
\quad &=\frac{\epsilon\min(i,n-i)(6+9\epsilon)}{n}.
\end{split}
\end{equation*}

\end{proof}
The second part of the error in (\ref{equation:contributing}) corresponds to the error associated with evaluating the empirical CDF $\hat{F}_{\hat{n}_1,(2)}$ with the query of the $\epsilon u$-approximate quantile summary $M(S_{(2)}^1,\ldots,S_{(2)}^L)$, instead of the empirical quantile function $\hat{F}^{-1}_{n,(2)}$.
\begin{theorem}
The term (B) in (\ref{equation:contributing}) can be bounded by,
$$
\abs*{\frac{\hat{n}_1}{n}\hat{F}_{\hat{n}_1,(2)}(\tilde{\tilde{F}}^{-1}_{n,(2)}(i/n))-\frac{\hat{n}_1}{n}\hat{F}_{\hat{n}_1,(2)}(\hat{F}^{-1}_{n,(2)}(i/n))} \leq \frac{\epsilon \min(i,n-i)}{n}
$$
\end{theorem}

\begin{proof}
This can be proved in the same way as Theorem 3 was in \cite{Gregory}, only with a $\pm \epsilon \min(i,n-i)$ error from the $\epsilon u$-approximate summaries rather than the $\pm \epsilon n$ error from the $\epsilon$-approximate summaries.
\end{proof}
Finally, the third part of the error in (\ref{equation:contributing}) corresponds to the error associated with using the cardinality of the set $\hat{I}$ in Sec. \ref{sec:querying}, $\hat{n}_1$, instead of the cardinality of the set $I$ in Sec. \ref{sec:empiricalcopulas}, $n_1$.
\begin{theorem}
The term (C) in (\ref{equation:contributing}) can be bounded by,
$$
\abs*{\frac{\hat{n}_1}{n}\hat{F}_{\hat{n}_1,(2)}(\hat{F}^{-1}_{n,(2)}(i/n))-\frac{n_1}{n}\hat{F}_{n_1,(2)}(\hat{F}^{-1}_{n,(2)}(i/n))} \leq \frac{\epsilon \min(i,n-i)}{n}
$$
\end{theorem}

\begin{proof}
This can be proved in the same way as Theorem 4 was in \cite{Gregory}, only with a $\pm \epsilon \min(i,n-i)$ error from the $\epsilon u$-approximate summaries rather than the $\pm \epsilon n$ error from the $\epsilon$-approximate summaries.
\end{proof}

\subsection{Space-memory}

\label{sec:spacememory}

A single $\epsilon u$-approximate quantile summary is of the worst case length $L=\mathcal{O}\big(\frac{\log ^2(\epsilon n)}{\epsilon}\big)$, since there is no benefit in using $u<1/(\epsilon n)$ \citep{Cormode}. Therefore the worst case length of the modified copula summary is $L^2=\mathcal{O}\big(\frac{\log^4(\epsilon n)}{\epsilon^2}\big)$. This bound is not saturated in all cases; for instance it was shown in \cite{Greenwald} that for some random streams the space-efficient summaries used in this paper can have size independent of $n$. The saturation of the worst case length of the $\epsilon u$-approximate summaries that make up the modified copula summary, given above, is discussed in \cite{Zhang}. For streams that are increasing in value for example, the space memory used by the $\epsilon u$-approximate quantile summaries is increasing with the length of the data stream. For the scope of this paper, this issue is not investigated further. However there are alternative approaches to relative error approximations of empirical quantile functions proposed in the literature; see \cite{Zhang} for more details on these.

It may be of interest to the reader to adapt this methodology to studying the tail dependence of higher dimensional data streams. This would be achieved by approximating higher dimensional empirical copulas, and then substituting these approximations into the tail dependence coefficient expressions in (\ref{equation:dependencefunction}) and (\ref{equation:upperdependencefunction}). Unfortunately a natural extension of the copula summary structure presented here to higher dimensions would incur an exponentially growing (with dimension) space-memory requirement \citep{Hershberger}. Turn to the work in \cite{Gregory} for an example of how one may utilise bivariate copula summaries to estimate higher dimensional empirical copulas via decompositions involving a sequence of bivariate empirical copulas \citep{Aas}, and alleviate this issue.


\section{Simulations and case-study}

\label{sec:numericalanalysis}

This section will provide some numerical demonstrations of the tail dependence coefficient estimation scheme in the streaming data regime, proposed in this paper. The theoretical analysis given throughout the paper is also numerically supported here. In Sec. \ref{sec:numericalapproximation}, the estimation of the tail dependence coefficients for random variables observed through data streams will be considered. Then in Sec. \ref{sec:numericalsizeandruntime}, the properties of the modified space-efficient copula summary (utilised for the estimation of the tail dependence coefficients) will be numerically investigated. Finally, a case-study of the Los Alamos National Laboratory (LANL) netflow data-set is considered.

\subsection{Approximation to the tail dependence coefficients}
\label{sec:numericalapproximation}

First, a demonstration of estimating both of the tail dependence coefficients between two random variables observed in a bivariate stream of data, will be presented. Consider the data stream $\big\{x_{(1)}^i,x_{(2)}^i\big\}_{i=1}^{n}$, where $n=3 \times 10^5$. Both components are randomly sampled from $\mathcal{N}(0,1)$ in one case, or $\text{Beta}(10, 1)$ in another. In both cases the components are correlated with Pearson's correlation $\rho=0.8$. The accuracy parameter $\epsilon$ takes the value of $0.1$, and the copula summary presented in this paper is constructed for five independent bivariate data streams for each distribution. We set $i=25$; this value is used to compute the estimate to both of the tail dependence coefficients using the copula summaries. These estimates are computed after every 5000'th element has been added to the data streams. The absolute error, away from the empirical lower and upper tail dependence coefficients of all data streams sampled from the Gaussian distribution, are shown over time in Figures \ref{fig:errorbias_lu_boxplot_2} and \ref{fig:errorbias_upper_lu_boxplot_2} respectively for the modified copula summary proposed in this paper. The error of the approximations to these tail dependence coefficients for all data streams sampled from the Beta distribution are also shown in Figures \ref{fig:errorbias_lu_boxplot_2_beta} and \ref{fig:errorbias_upper_lu_boxplot_2_beta}. Visible in all plots is the theoretical stream-length invariant bound presented in (\ref{equation:taildependencebound}) and (\ref{equation:uppertaildependencebound}). The error for the upper tail dependence coefficient is slightly more than that of the lower tail dependence coefficient in the case of the data streams sampled from the Gaussian distribution, and vice-versa for the data streams sampled from the Beta distribution.


\begin{figure}[!htb]
\centering
\minipage{0.47\textwidth}
\centering
  \includegraphics[width=\linewidth]{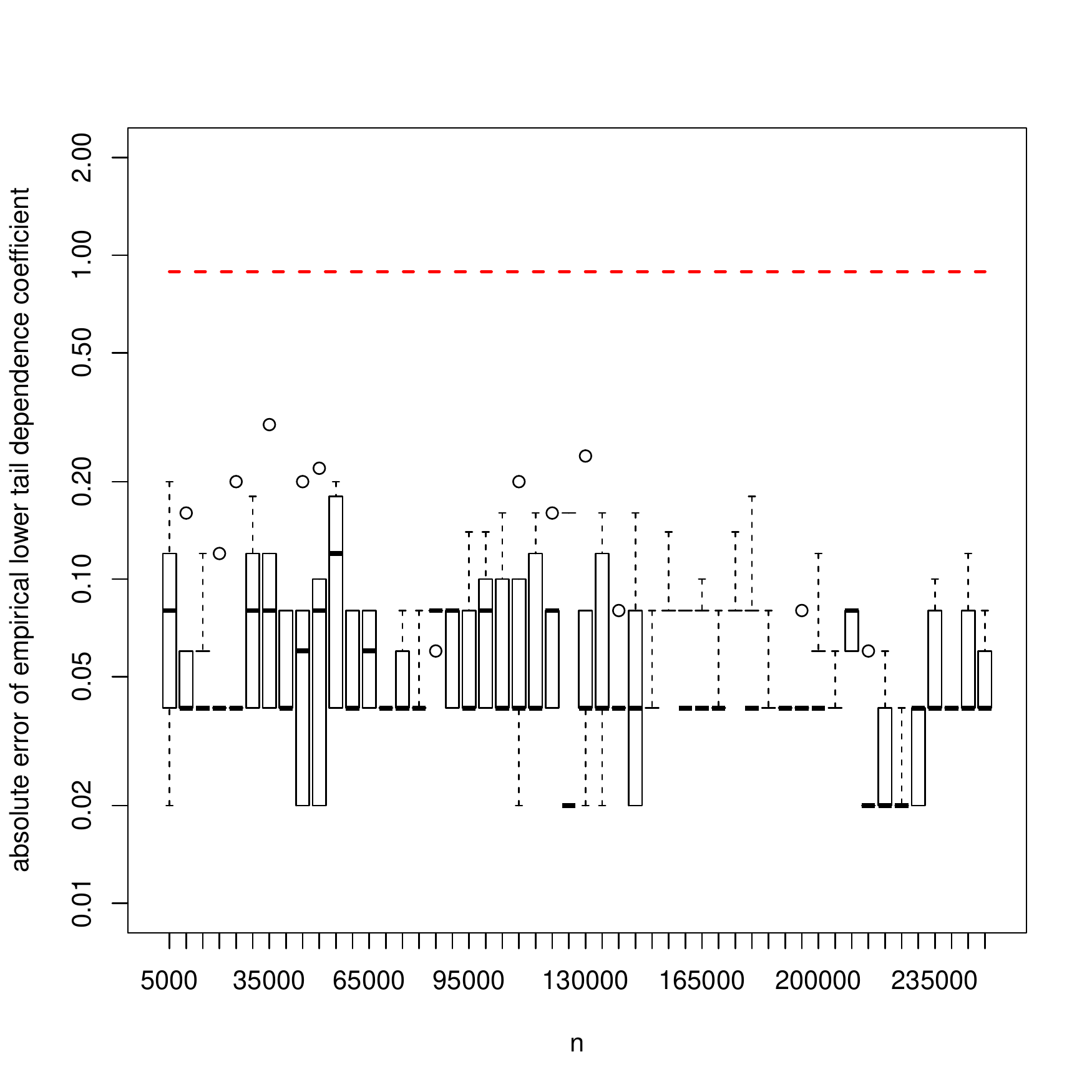}
  \caption{Boxplots of absolute error of the lower tail dependence coefficient approximations using the modified copula summary, with $\epsilon=0.1$, based on five independent data streams sampled from a Gaussian distribution. The red dashed line shows the theoretical bound in (\ref{equation:taildependencebound}).}\label{fig:errorbias_lu_boxplot_2}
\endminipage\hfill
\hspace{-5mm}
\minipage{0.47\textwidth}
\centering
  \includegraphics[width=\linewidth]{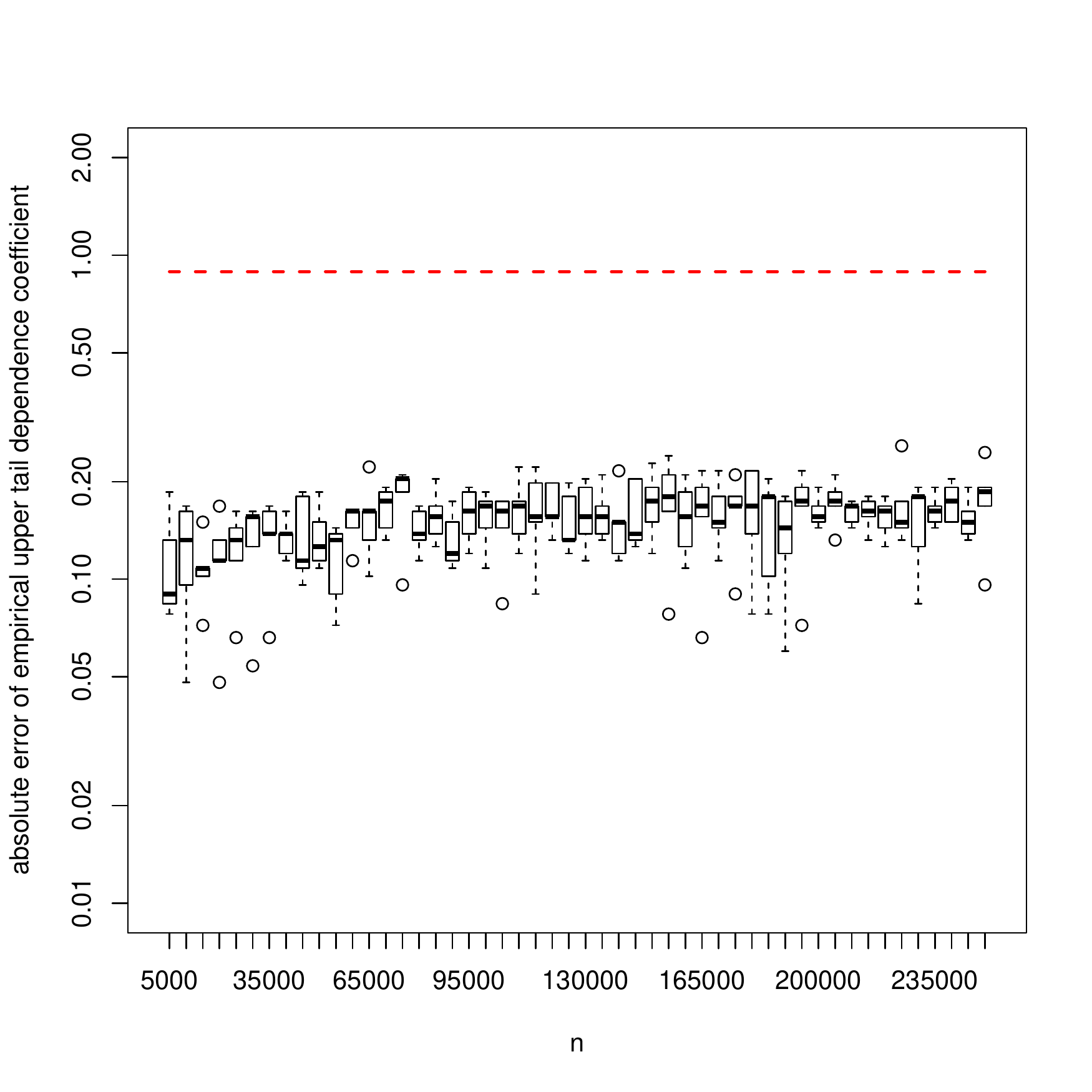}
  \caption{Boxplots of absolute error of the upper tail dependence coefficient approximations using the modified copula summary, with $\epsilon=0.1$, based on five independent data streams sampled from a Gaussian distribution. The red dashed line shows the theoretical bound in (\ref{equation:uppertaildependencebound}).}\label{fig:errorbias_upper_lu_boxplot_2}
\endminipage\hfill
\end{figure}

\begin{figure}[!htb]
\centering
\minipage{0.47\textwidth}
\centering
  \includegraphics[width=\linewidth]{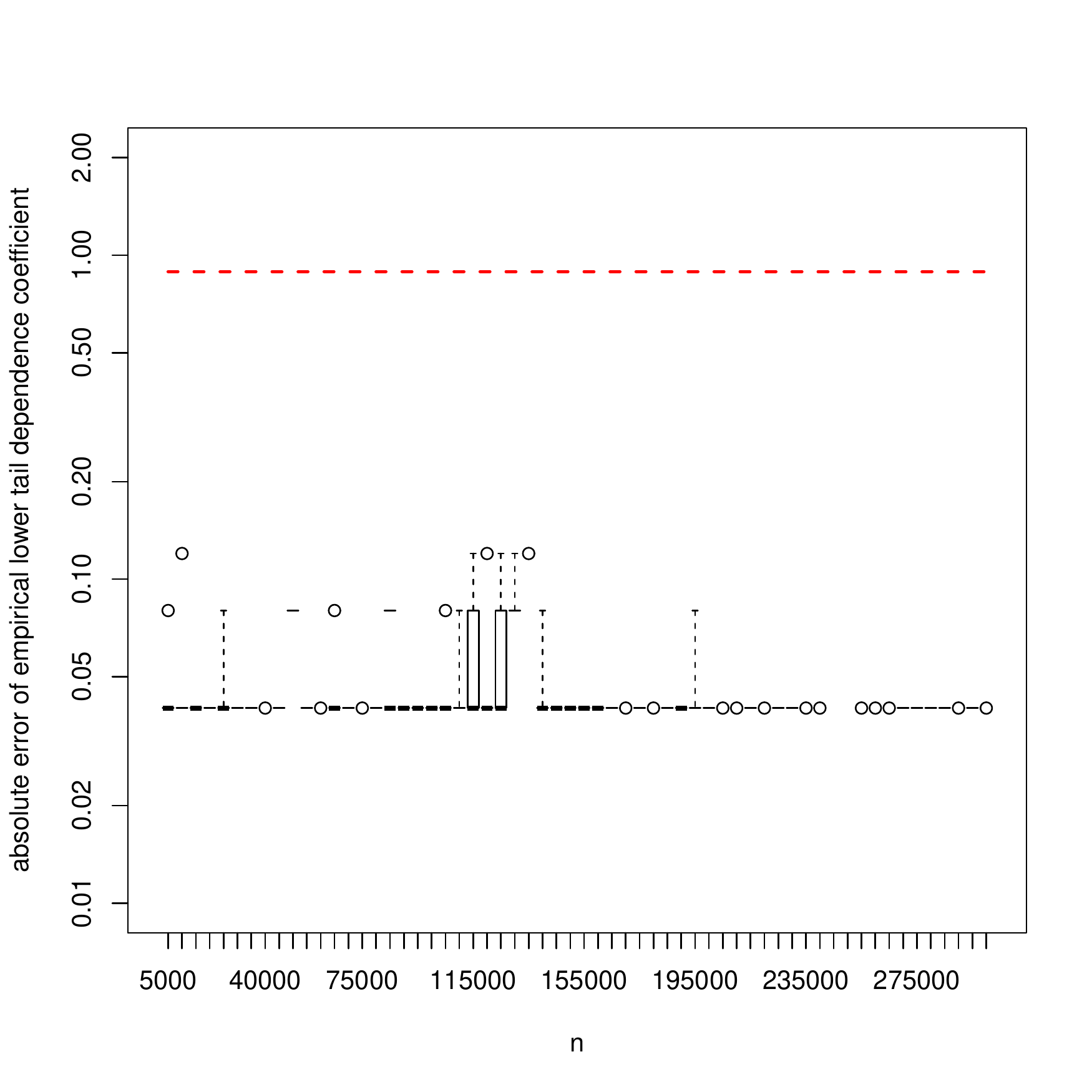}
  \caption{Boxplots of absolute error of the lower tail dependence coefficient approximations using the modified copula summary, with $\epsilon=0.1$, based on five independent data streams sampled from a Beta distribution. The red dashed line shows the theoretical bound in (\ref{equation:taildependencebound}).}\label{fig:errorbias_lu_boxplot_2_beta}
\endminipage\hfill
\hspace{-5mm}
\minipage{0.47\textwidth}
\centering
  \includegraphics[width=\linewidth]{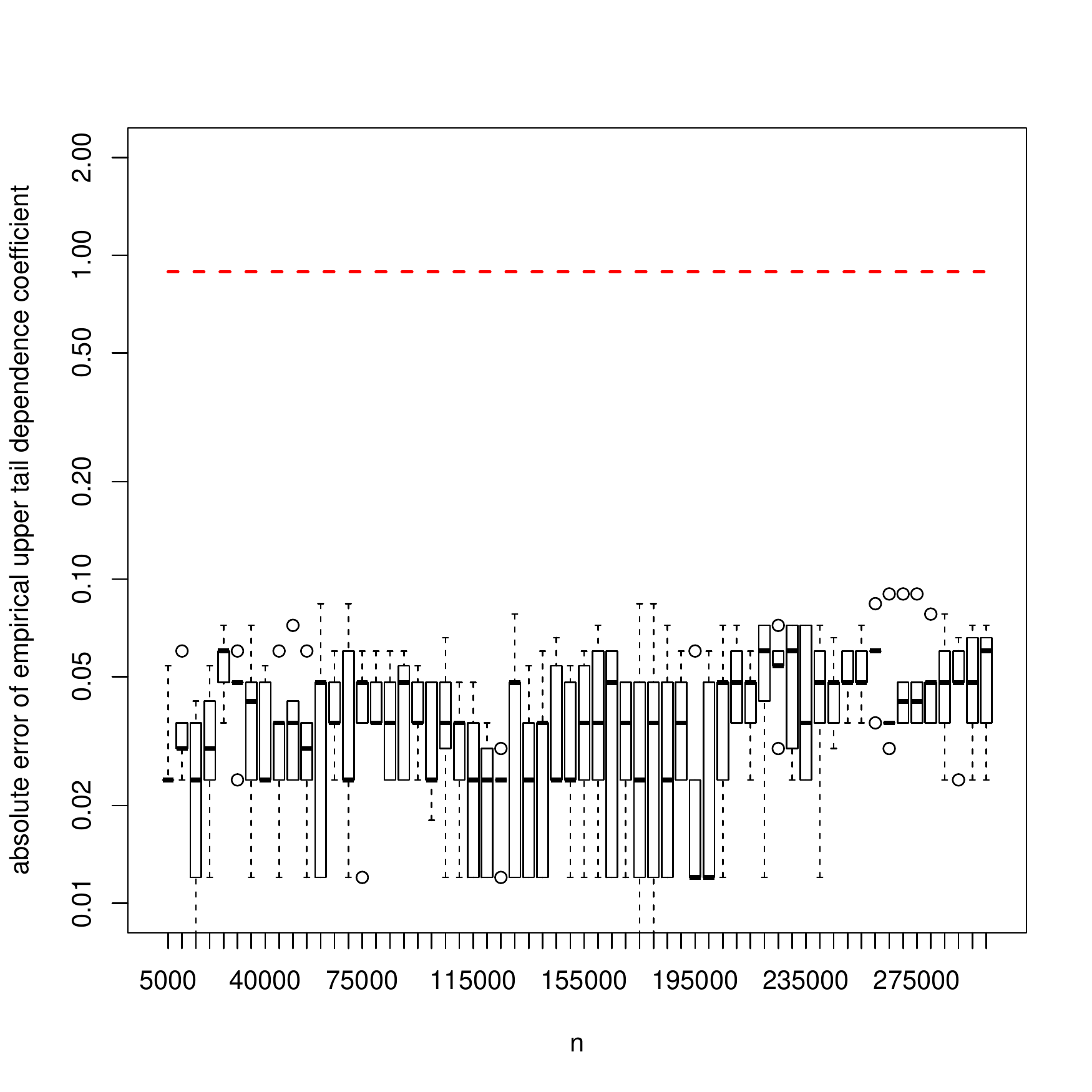}
  \caption{Boxplots of absolute error of the upper tail dependence coefficient approximations using the modified copula summary, with $\epsilon=0.1$, based on five independent data streams sampled from a Beta distribution. The red dashed line shows the theoretical bound in (\ref{equation:uppertaildependencebound}).}\label{fig:errorbias_upper_lu_boxplot_2_beta}
\endminipage\hfill
\end{figure}

\subsection{Properties of the copula summary}
\label{sec:numericalsizeandruntime}

Next the properties of the modified copula summary, including it's space-efficiency and implementation runtime, will be numerically demonstrated. Consider the modified copula summary presented in this paper, with $\epsilon=0.1$, used in the previous numerical experiment (where five independent bivariate data streams are sampled from both a Gaussian distribution and a Beta distribution). The absolute error of the approximation from this copula summary over time and all data streams sampled from the Gaussian distribution, for the evaluation points $(u_1,u_2)=(0.7,0.7)$ and $(u_1,u_2)=(0.02,0.02)$, are shown in Figures \ref{fig:error_lu_boxplot_2} and \ref{fig:error1_lu_boxplot_2} respectively. The same error, only for the data streams sampled from the Beta distribution, are shown in Figures \ref{fig:error_lu_boxplot_2_beta} and \ref{fig:error1_lu_boxplot_2_beta}. These show that, as with the standard copula summary presented in \cite{Gregory}, the error of the empirical copula approximation is bounded by a constant. However now, in the case where the evaluation point $(0.7, 0.7)$ is not in the tails of each marginal, the approximation has a higher bound (and therefore exhibits greater numerical error) than in the case where the evaluation point $(0.02, 0.02)$ is in the tails. This is in line with the theoretical analysis presented in (\ref{equation:modifiedcopulabound}). Next, the implementation runtime and space-efficiency of the modified copula summary, over all data streams sampled from the Gaussian distribution, is demonstrated. Figure \ref{fig:runtime_lu_boxplot_2} shows the runtime (in seconds) of an iteration of the copula summary algorithm after every 5000'th element is added to the streams. Occassional peaks are due to the combination operation in Sec. \ref{sec:combining} being implemented. Similarly, Figure \ref{fig:size_summary_ratio_lu_boxplot_2} shows the \textit{size ratio} of the modified copula summary to the entire streams:
$$
\text{size ratio} = \frac{\text{size of copula summary}}{\text{size of entire stream}} ,
$$
after every 5000'th element is added to the streams. This shows the increasing space-efficiency of the copula summary proposed in this paper, used for approximations to the empirical tail dependence coefficients, as the stream length increases.
At the end of the data stream the size of the entire stream is 5440520 bytes, 39 times the size of the copula summary.


\begin{figure}[!htb]
\centering
\minipage{0.47\textwidth}
\centering
\includegraphics[width=\linewidth]{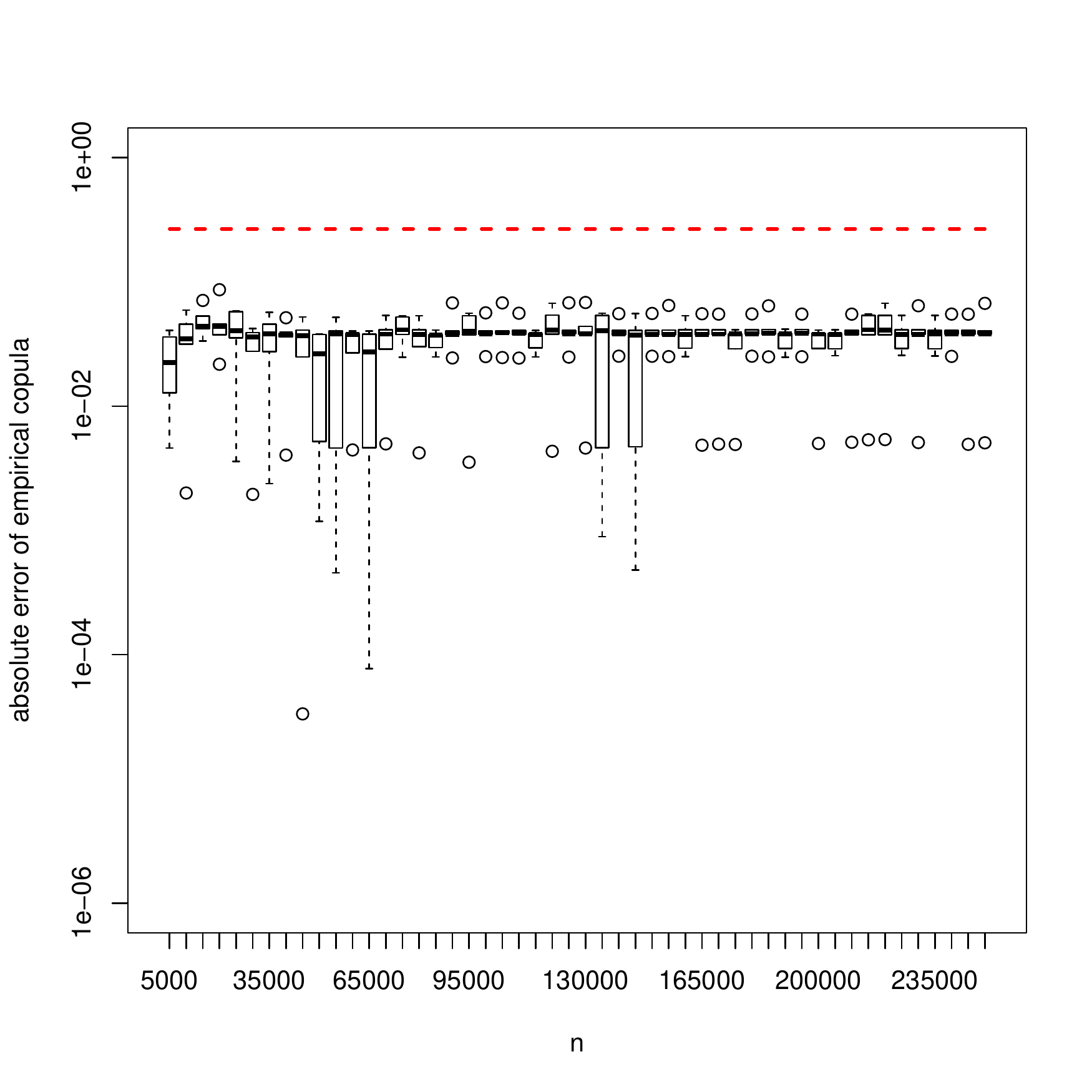}
\caption{Boxplots of absolute error, over time, of the empirical copula approximation $\hat{C}(0.7, 0.7)$ using the modified copula summary, with $\epsilon=0.1$, based on five independent data streams sampled from a Gaussian distribution. The bound in (\ref{equation:modifiedcopulabound}) is shown by the red line.}\label{fig:error_lu_boxplot_2}
\endminipage\hfill
\hspace{-5mm}
\minipage{0.47\textwidth}
\centering
\includegraphics[width=\linewidth]{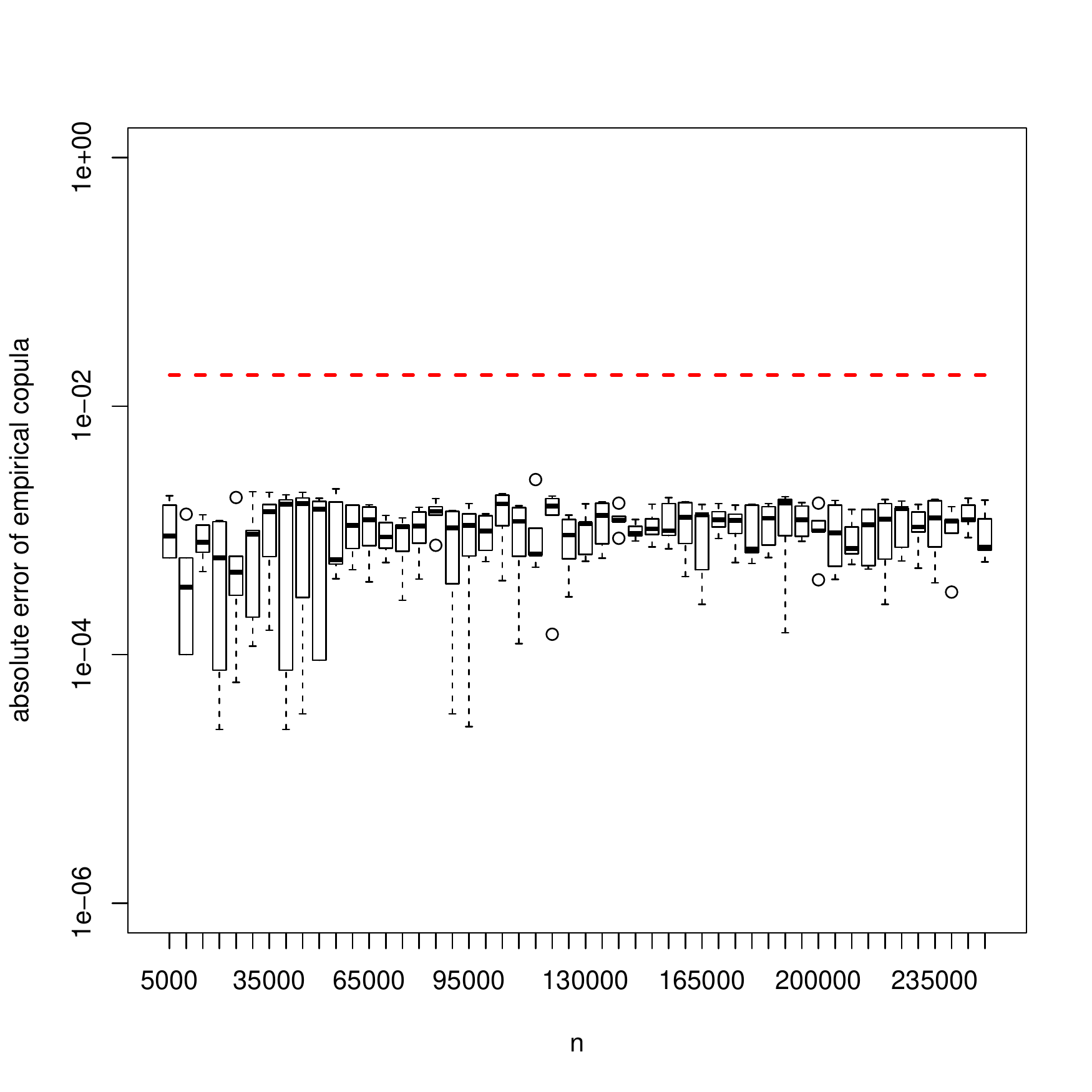}
\caption{Boxplots of absolute error, over time, of the empirical copula approximation $\hat{C}(0.02, 0.02)$ using the modified copula summary, with $\epsilon=0.1$, based on five independent data streams sampled from a Gaussian distribution. The bound in (\ref{equation:modifiedcopulabound}) is shown by the red line.}\label{fig:error1_lu_boxplot_2}
\endminipage\hfill
\end{figure}

\begin{figure}[!htb]
\centering
\minipage{0.47\textwidth}
\centering
\includegraphics[width=\linewidth]{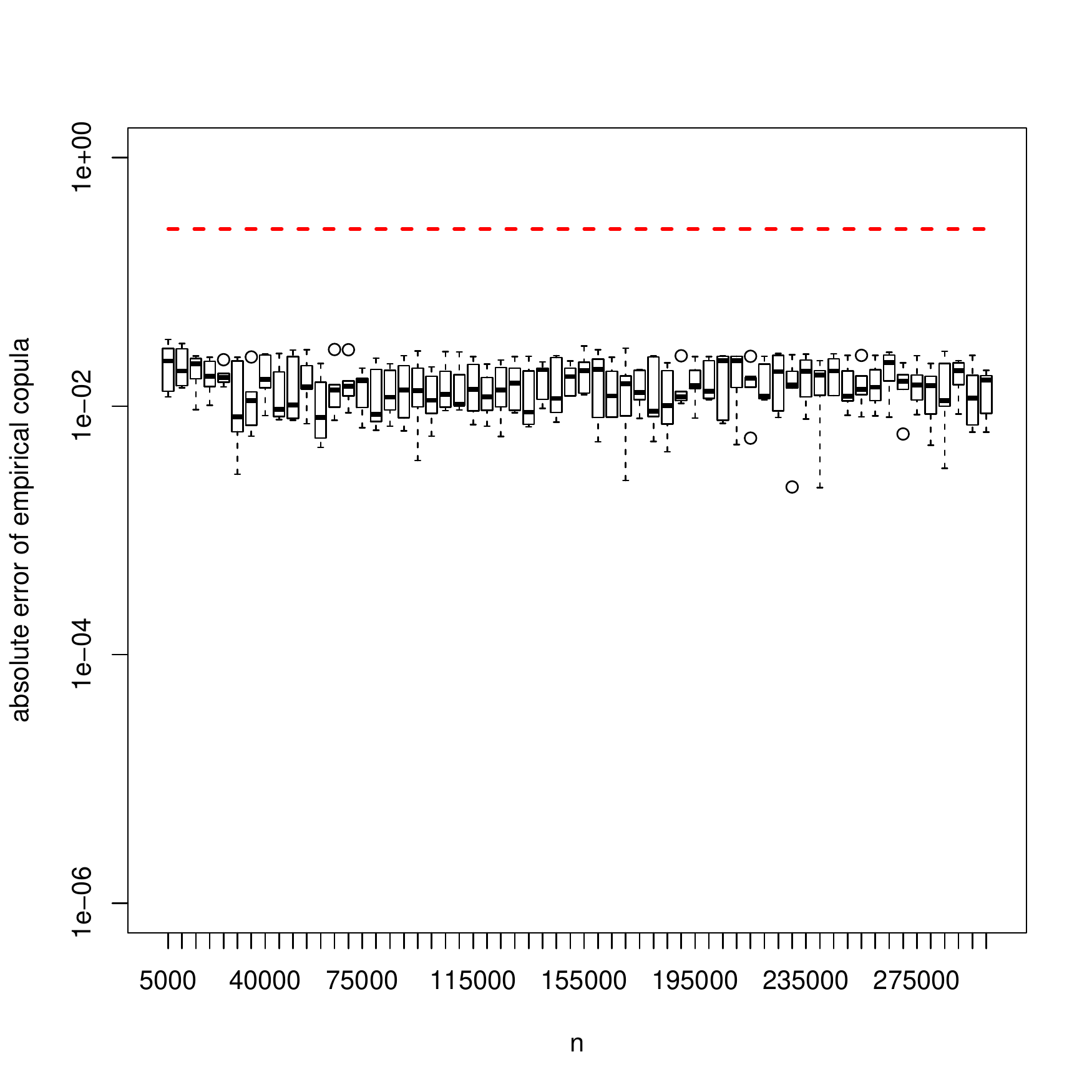}
\caption{Boxplots of absolute error, over time, of the empirical copula approximation $\hat{C}(0.7, 0.7)$ using the modified copula summary, with $\epsilon=0.1$, based on five independent data streams sampled from a Beta distribution. The bound in (\ref{equation:modifiedcopulabound}) is shown by the red line.}\label{fig:error_lu_boxplot_2_beta}
\endminipage\hfill
\hspace{-5mm}
\minipage{0.47\textwidth}
\centering
\includegraphics[width=\linewidth]{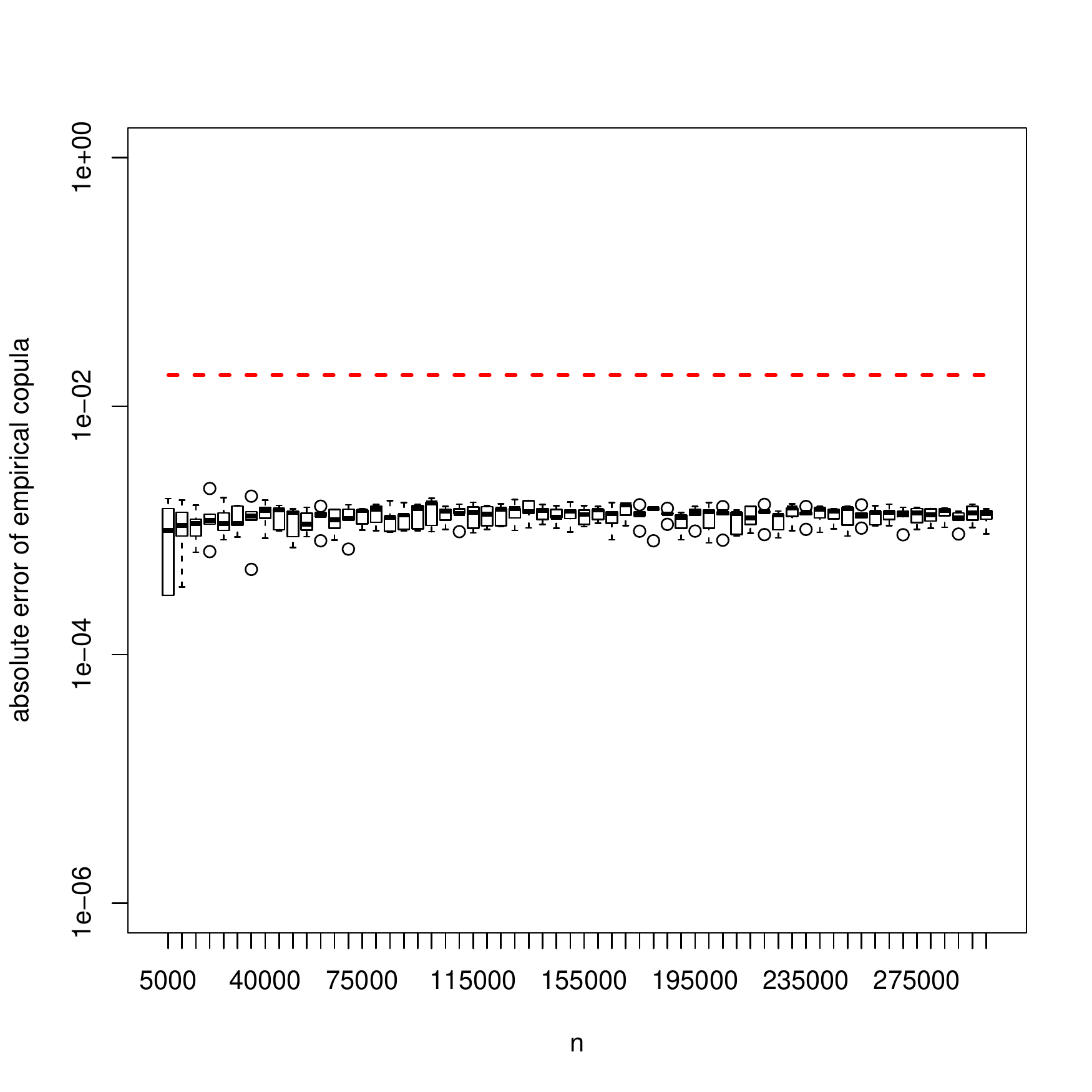}
\caption{Boxplots of absolute error, over time, of the empirical copula approximation $\hat{C}(0.02, 0.02)$ using the modified copula summary, with $\epsilon=0.1$, based on five independent data streams sampled from a Beta distribution. The bound in (\ref{equation:modifiedcopulabound}) is shown by the red line.}\label{fig:error1_lu_boxplot_2_beta}
\endminipage\hfill
\end{figure}

\begin{figure}[!htb]
\centering
\minipage{1\textwidth}
\centering
  \includegraphics[width=0.8\linewidth]{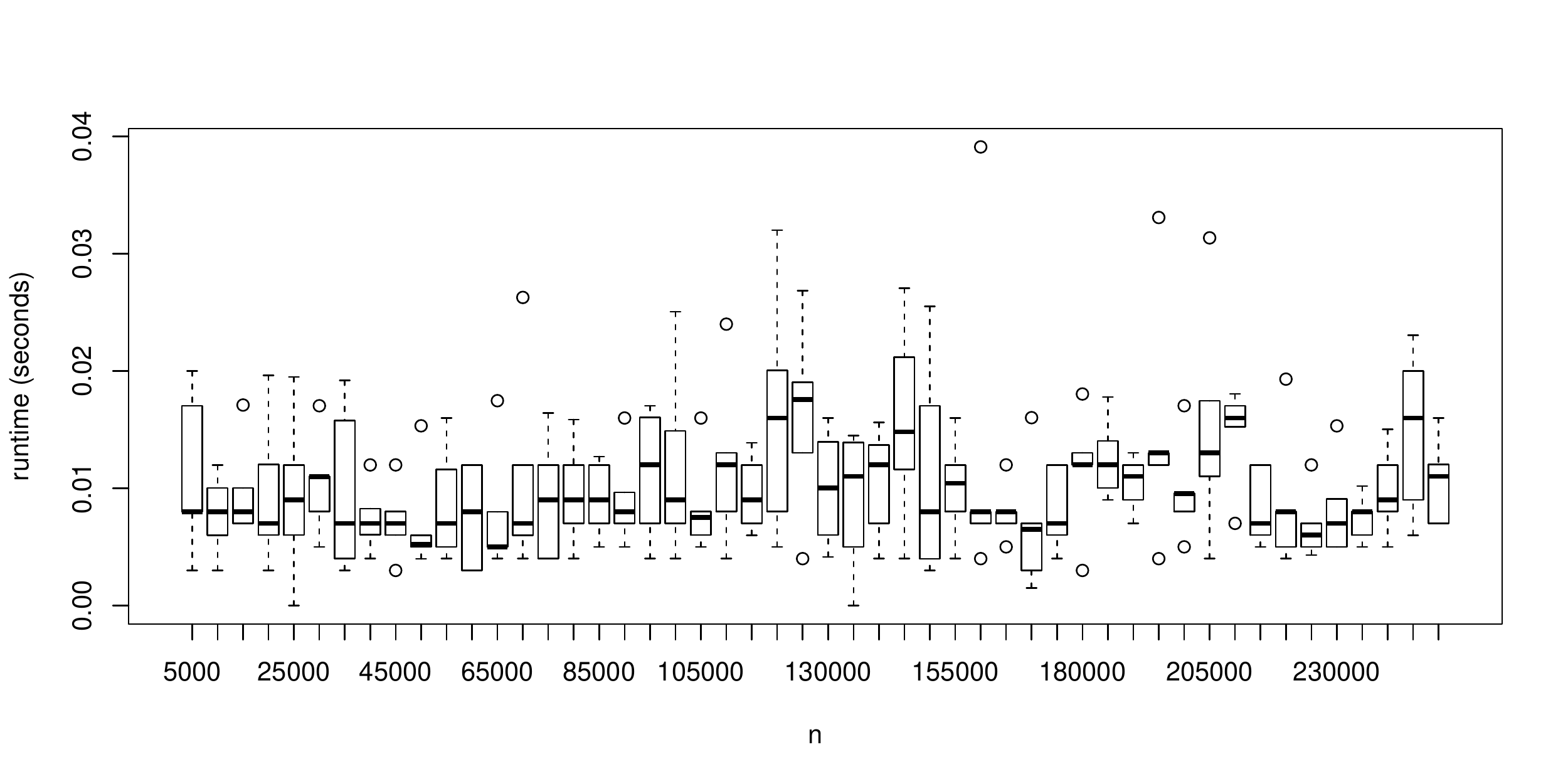}
  \caption{Boxplot of the implementation runtimes (in seconds) of the modified copula summary, with $\epsilon=0.1$, based on five independent data streams sampled from a Gaussian distribution.}\label{fig:runtime_lu_boxplot_2}
\endminipage\vfill
\minipage{1\textwidth}
\centering
  \includegraphics[width=0.8\linewidth]{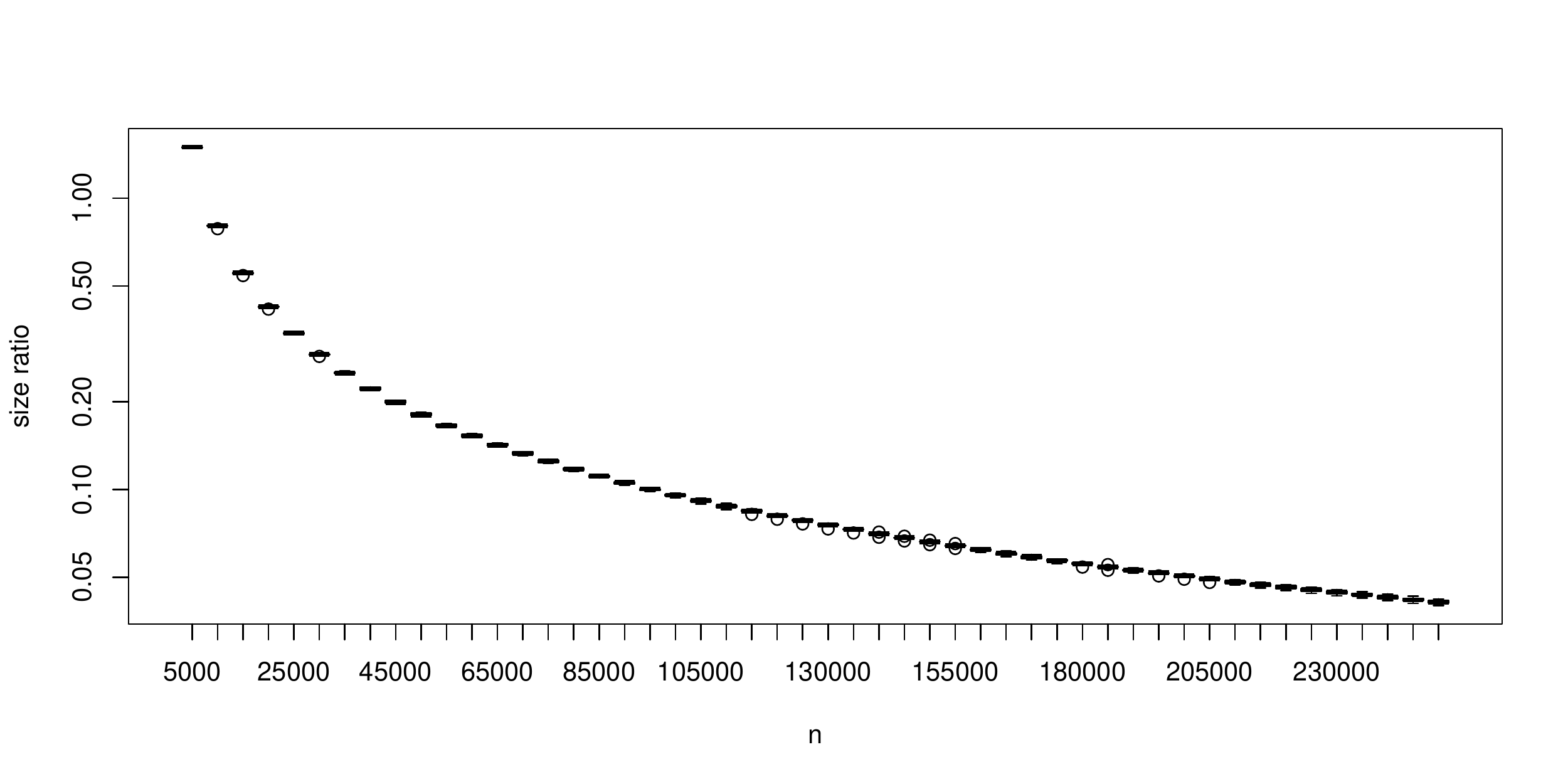}
  \caption{Boxplot of the size (in bytes) ratio of the modified copula summary to the entire stream, over time, with $\epsilon=0.1$, based on five independent data streams sampled from a Gaussian distribution.}\label{fig:size_summary_ratio_lu_boxplot_2}
\endminipage
\end{figure}

Finally, the effect of $\epsilon$ on the approximation error of the modified copula summary presented in this paper is explored. The absolute error of the approximation from the copula summary at the evaluation points $(u_1,u_2)=(0.7,0.7)$ and $(u_1,u_2)=(0.02,0.02)$, over five independent data streams of length 30000 (sampled from the aforementioned Gaussian distribution), is shown for a number of different $\epsilon$ values in Figure \ref{fig:error_copula_epsilon_boxplot}. The error was computed after every element has been added to the data streams. As found for a fixed value of $\epsilon$ in Figures \ref{fig:error_lu_boxplot_2} and \ref{fig:error1_lu_boxplot_2}, the copula approximation error is lower for the evaluation point in the tail, $(0.02,0.02)$. In addition to this, as $\epsilon$ decreases, the approximation error at both evaluation points decreases respectively too; this behaviour is described in the theoretical bound presented in (\ref{equation:modifiedcopulabound}).

\begin{figure}[!htb]
\centering
  \includegraphics[width=1\linewidth]{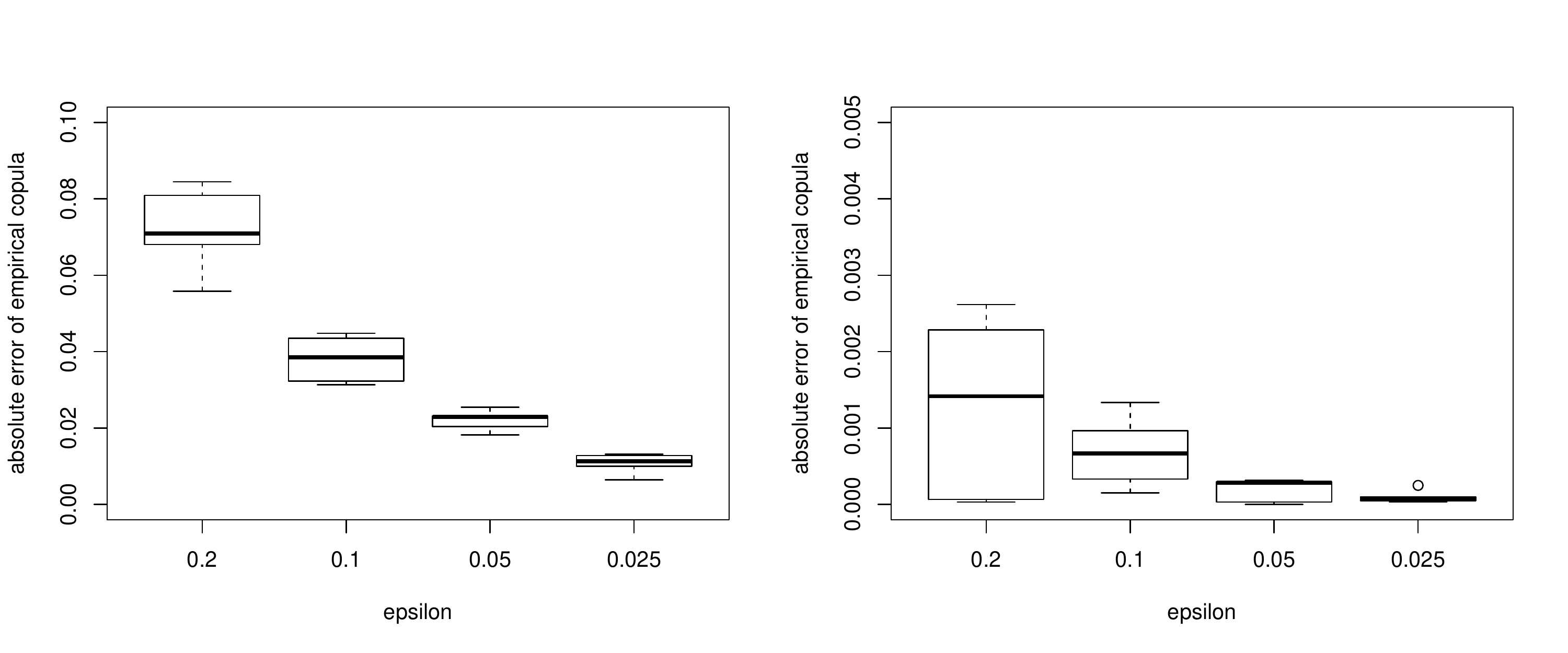}
  \caption{Boxplots of absolute error of the empirical copula approximation of $\hat{C}(0.7,0.7)$ (left) and $\hat{C}(0.02,0.02)$ (right) using the modified copula summary, with varying $\epsilon$, based on five independent data streams each. The error was computed after every element has been added to the data streams.}\label{fig:error_copula_epsilon_boxplot}
\end{figure}

\subsection{LANL netflow data}\label{data_analysis}

This section applies the proposed methodology to a case-study of cyber-netflow traffic data consisting of records of computer network communication, available freely from Los Alamos National Laboratory (LANL, https://csr.lanl.gov/data/2017.html). This provides an example of streaming data with great importance in the field of cyber-security \citep{Adams2016}. Various aspects of this data are throughly described in \cite{Turcotte2017}. It has been shown that flow based techniques have a number of computational advantages and are successful in detecting a variety of malicious network behaviors \citep{Sperotto2010}. Daily netflow data is available for 89 days (indexed as Day 2 to Day 91, starting with Day 2). Each flow records  an aggregate summary of a bi-directional network communication between any two of approximately 60000 devices of the LANL network \citep{Turcotte2017}. The aggregate summary consists of various records on the communication for every second of the day (given in epoch time format). We consider the variable {\it SrcPackets}, the number of bytes the {\it SrcDevice} sent during a certain communication event, where {\it SrcDevice} is the device that likely initiated the event. 

Modeling the records of communication events in a fixed time window is a common approach to understand normal behaviour of a cyber network, and in turn helps in detecting unusual behaviour \citep{Evangelou2016}. This type of behaviour could be detected by studying outlying extreme values of either variable mentioned above, however these particular values might be highly correlated with extrema from the other variable. By knowing the tail dependence of the variables through the tail dependence coefficients, one could reason about how likely extrema occuring for both variables at the same time signifies unusual behaviour. In this application, we consider the above data in the streaming sense over a fixed time window of one second (lowest time resolution available in the netflow data-set). Then for a given day, we consider the bivariate data stream $\big\{x_{(1)}^i,x_{(2)}^i\big\}$, where $x_{(1)}^i$ is the total volume of bytes transferred during events that initiated at time $i$ and $x_{(2)}^i$ is the total number of active edges in the network at time $i$, $i=1,\ldots,86400$. In the following analysis of the netflow data, we use log transform values of both the variables. Figure~\ref{ext-index-1} shows scatter plots of $\big\{x_{(1)}^i,x_{(2)}^i\big\}$, $i=1,\ldots,86400$, for two randomly selected days, namely Day 3 and Day 10.  The points in the red highlight those lying below the 0.005th sample quantile of \textit{both} marginals, whereas the darkened points highlight those lying in the 0.005th sample quantile of \textit{either} marginal. In both of the days, the pattern indicates a positive co-movement in the lower tails.

Figure \ref{approx_ltdc_over_i_netflow} shows the empirical lower tail dependence coefficient between the total volume of bytes transferred and the number of active edges over Day 10, for different values of $i$. As one can see, there is significant tail dependence demonstrated by a non-zero coefficient as $i$ tends to 0. Also shown in this Figure is the approximation to the lower tail dependence coefficient computed using the proposed copula summary in this paper, with $\epsilon=0.075$. The copula summary is constructed sequentially after each element is added to the data stream on each second of Day 10, in the manner described in Sec. \ref{sec:modificationsummary}. This represents the summary that would be maintained if this data was to be streamed in real-time to a netflow traffic analyzer. The approximation shown in Figure \ref{approx_ltdc_over_i_netflow} is computed after the entire data stream from Day 10 has been inserted into the copula summary; it shows the significant non-zero tail dependence coefficient well for all $i$ values considered. The average time for each element to be inserted into the copula summary was just 0.03 seconds; therefore this algorithm would be sufficiently fast to match the data acquisition rate (one second). The total number of elements stored in the copula summary was 26448, which is much less than the number of elements in the data stream (172800), and results in a size ratio of 0.15. Therefore this approximation of the lower tail dependence coefficient is space-efficient relative to storing the entire data stream.

\begin{figure}[]\centering\includegraphics[height=2.5 in]{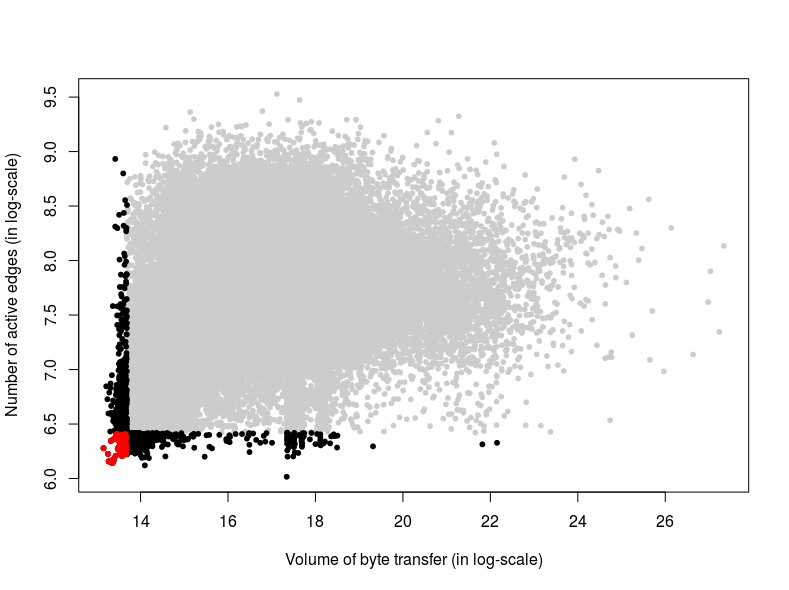}\includegraphics[height=2.5 in]{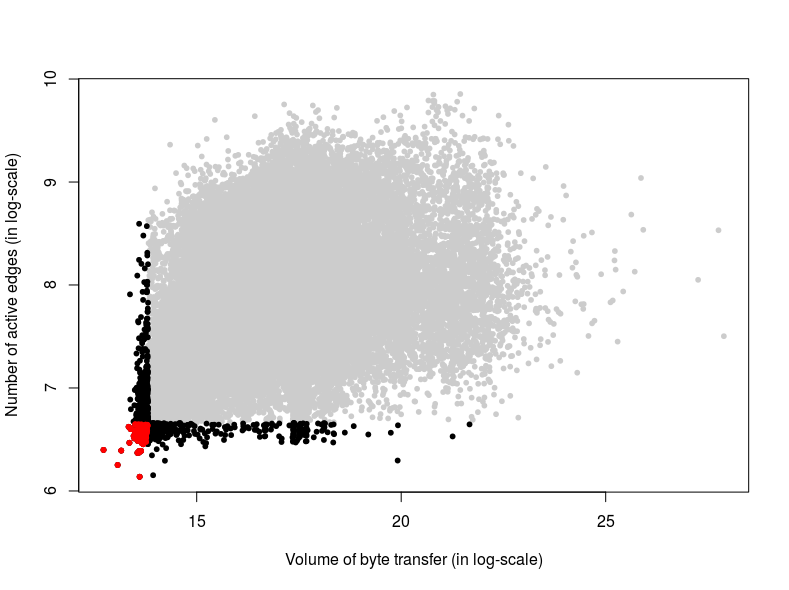} \caption[]{A scatter plot of the total volume of bytes transferred and number of active edges for Day 3 (left) and Day 10 (right) of the LANL netflow data. The darker points correspond to all points in the 0.005th quantile of \textit{each} marginal, and the red points correspond to all points in the 0.005th quantile of \textit{both} marginals.}\label{ext-index-1}\end{figure}
 

\begin{figure}[!htb]
\centering
  \includegraphics[width=130mm]{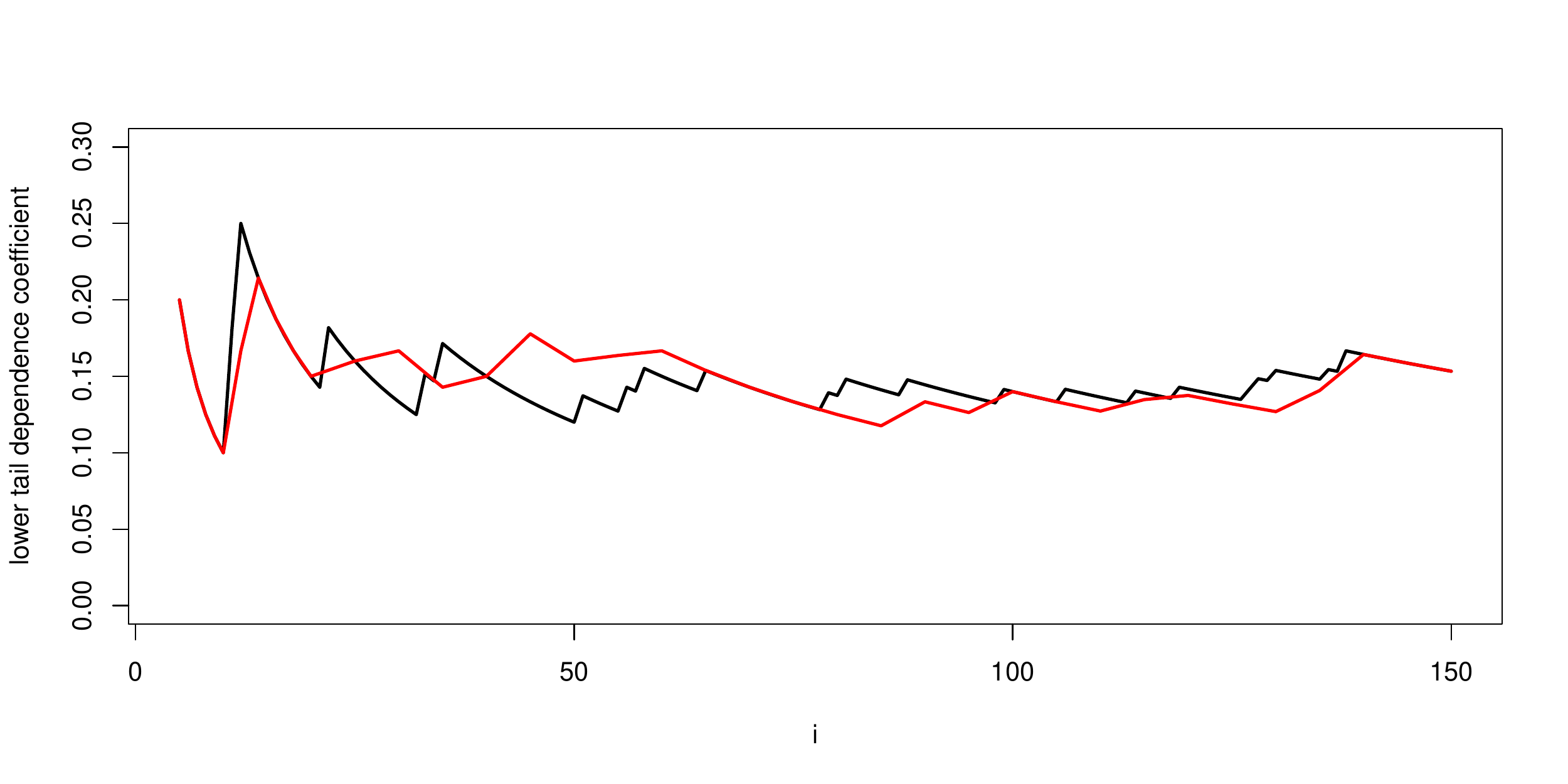}
	\caption[]{The lower tail dependence coefficient (in black) between the total volume of bytes transferred and number of active edges for Day 10 of the LANL netflow data, over different values of $i$. The approximation of this coefficient in (\ref{equation:newapproximationtaildependence}) over the same values of $i$ is also shown (in red).}\label{approx_ltdc_over_i_netflow}
\end{figure}

To see the impact of the condition in (\ref{equation:empiricalrankcondition}) on the structure of the copula summary applied to approximating the lower tail dependence coefficient for this netflow traffic data, Figure \ref{fig:g_plot_netflow} shows the values of $g_{(1)}^i$, for $i=1,\ldots,L$, defined in (\ref{equation:gis}), after all elements have been added to the data stream. Recall these values are the $g$ values within the tuples making up the $\epsilon u$-approximate quantile summary $S_{(1)}$ in the copula summary. In addition to this, Figure \ref{fig:subsummary_length_netflow} shows the values of $L_i$ (the subsummary $S_{(2)}^{i}$ lengths), for $i=1,\ldots,L$. These show a similar pattern to the values of $g_{(1)}^i$; where there is only very few elements being represented by the $i$'th tuple in $S_{(1)}$, there is also only very few elements to have entered or been merged into the corresponding subsummary $S_{(2)}^i$. This illustrated behaviour of the copula summary, that makes up the approximations to the lower tail dependence coefficient of the netflow traffic data, demonstrates the benefits of the modifications presented in this paper.

\begin{figure}[!htb]
\centering
\minipage{0.47\textwidth}
\centering
  \includegraphics[width=\linewidth]{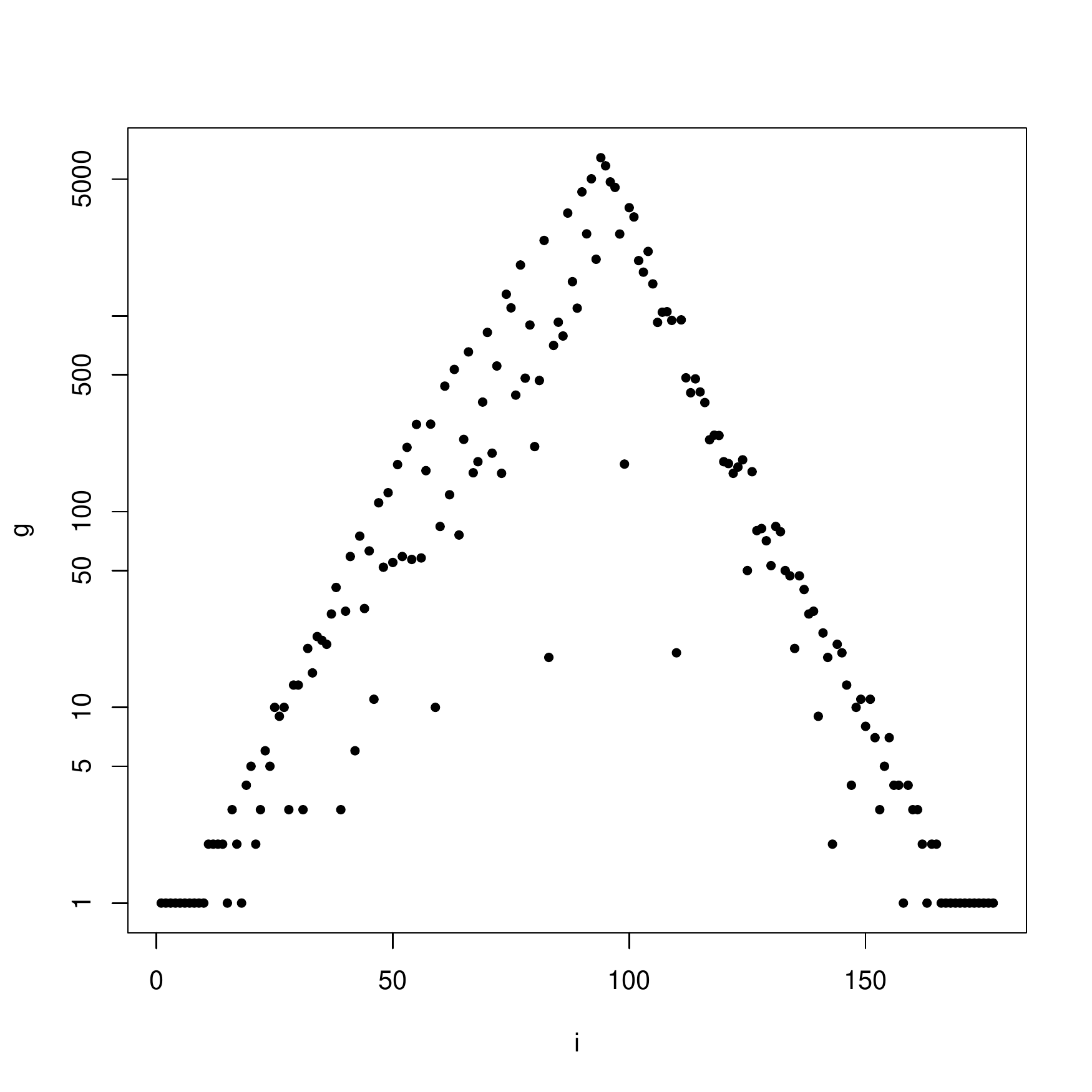}
  \caption{The values of $g_{(1)}^i$ for $i=1,\ldots,L$, in the modified copula summary constructed over Day 10 of the LANL netflow data.}\label{fig:g_plot_netflow}
\endminipage\hfill
\minipage{0.47\textwidth}
\centering
  \includegraphics[width=\linewidth]{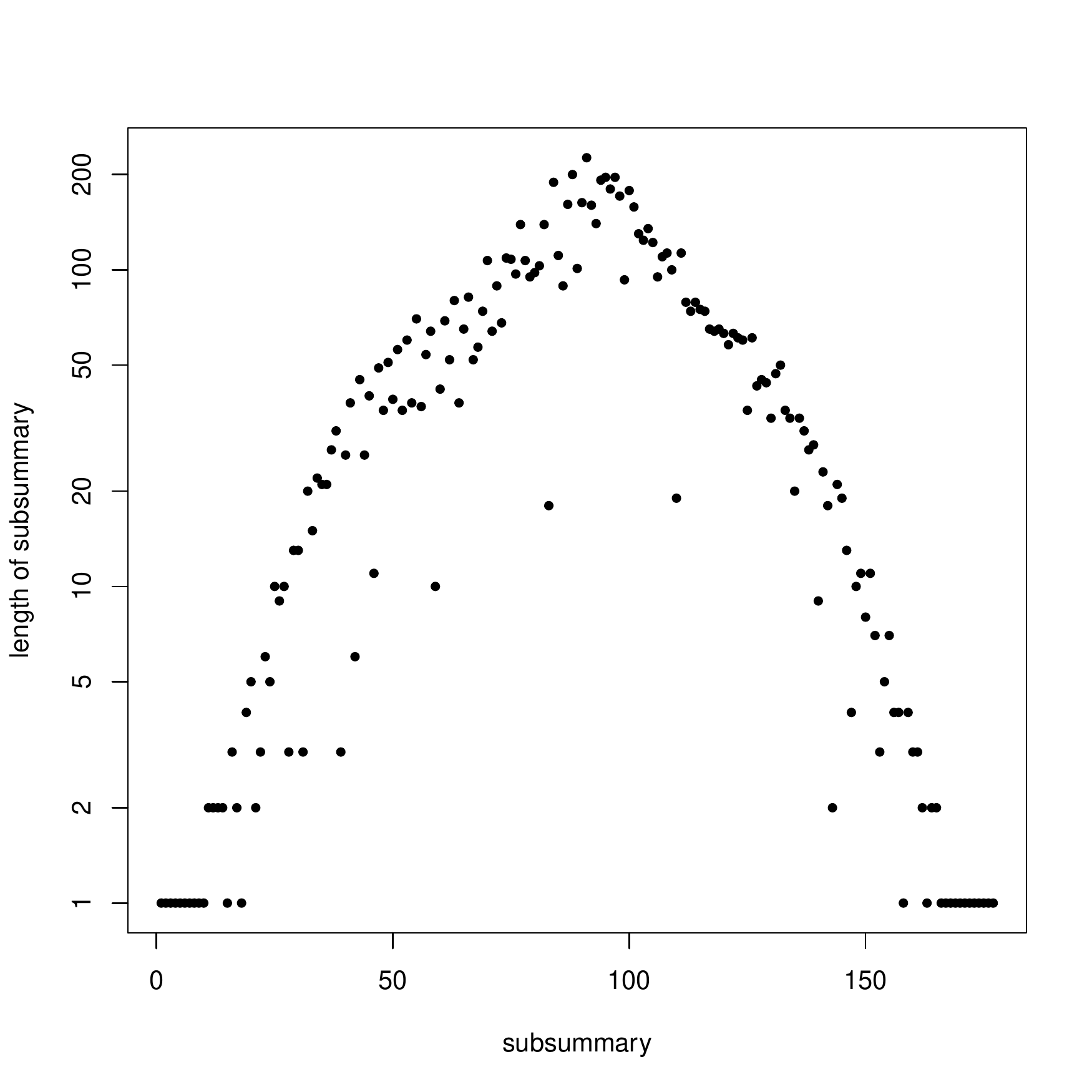}
  \caption{The values of (subsummary lengths) $L_i$ for $i=1,\ldots,L$, in the modified copula summary constructed over Day 10 of the LANL netflow data.}\label{fig:subsummary_length_netflow}
\endminipage
\end{figure}

\section{Summary and conclusion}

This article has presented a stream-length invariant bounded and space-efficient approximation to empirical tail dependence coefficients for bivariate streaming data. This regime of data means that an indefinite set of data can't be stored in its entirety, and analyzes of the data (such as modelling the tail dependence) need to be updated on-the-fly. The approximation presented in this paper is implemented via the use of a modified copula summary; the standard copula summary was introduced recently in \cite{Gregory}. The modification, first introduced in the context of quantile estimation \citep{Cormode}, allows the error of the approximation to be refined in the tails of the copula marginals, and therefore it does not grow linearly with the number of elements in the data stream (such as the case is when using the standard copula summary).

The methodology presented in this paper is for use with bivariate streaming data. An example of how the copula summaries, that make up the approximation to empirical tail dependence coefficients in this paper, could be extended to higher dimensional data streams was presented in \cite{Gregory}. Due to the wide-range of industries that now use streaming data, developing the techniques surrounding dependence modelling for this type of data is important; this paper continues this line of work. A relevant case-study of such an industry (cyber-security) is considered at the end of this paper; the proposed methodology is employed to space-efficiently capture lower tail dependence in a netflow data-set from the Los Alamos National Laboratory.  



\begin{appendix}

\section{Proof of bound on error for $\epsilon u$-approximate quantile summary approximations}
\label{sec:proofofnewinvariant}

This section shows that the condition in (\ref{equation:empiricalrankcondition}),
$$
r_{max,Q}(v_{i+1})-r_{min,Q}(v_i) \leq 2\epsilon \min\left( r_{min,Q}(v_i), n-r_{max,Q}(v_{i+1})\right),
$$
guarantees that a $\epsilon u$-approximate summary $Q$ can return an approximation $\tilde{x}_{(k)}^j$, to $\hat{F}_{n,(k)}(u)=\tilde{x}_{(k)}^{\ceil{un}}$ for $k=1,2$, where $j \in [\ceil{un}-\epsilon\min(u,(1-u)) n,\ceil{un}+\epsilon\min(u,(1-u)) n]$. This is a generalisation of the analysis in \cite{Cormode}. First, we must show that the insertion and combining operations in Sec. \ref{sec:inserting} and \ref{sec:combining} do not alter the bound in (\ref{equation:empiricalrankcondition}). Start by considering the insertion operation outlined in Sec. \ref{sec:inserting}. When this operation is carried out and the added tuple gets input after $v_{i+1}$, the first term in the minimum is unaffected. The second term in the minimum in this case is also unaffected, since $r_{max,Q}(v_{i+1})$ increases by 1 ($g=1$) but so does $n$; this means that the bound is still satisfied. If the added tuple gets input before $v_{i}$, then the first term in the minimum increases by 1; this also means that the bound is still satisfied. The second term in the minimum in this case is again unaffected since both $r_{max,Q}(v_{i+1})$ and $n$ increases by 1. Therefore the insertion operation does not affect the bound in (\ref{equation:empiricalrankcondition}). Next consider the combining operation in Sec. \ref{sec:combining}. Clearly this operation only combines tuples in $Q$ when the condition in (\ref{equation:empiricalrankcondition}) is satisfied and therefore this operation does not alter the required bound.

Finally, we now show that $Q$ can be queried to return a value $j$ where $j \in [\ceil{un}-\epsilon\min(u,(1-u)) n,\ceil{un}+\epsilon\min(u,(1-u)) n]$. Assume for simplicity that $\ceil{un}=un$, however the derivation below can simply be adjusted to relax this assumption. Let $r_i=r_{max,Q}(v_{i+1})-r_{min,Q}(v_i)$ and also let $i$ be the smallest index that satisfies
\begin{equation}
r_i+g^i+\Delta^i > un+\epsilon\min\left(un,(1-u)n\right).
\label{equation:greaterthanrcondition}
\end{equation}
Therefore, we have $r_{i-1}+g^{i-1}+\Delta^{i-1} \leq \min\left((1+\epsilon)un, \epsilon n + (1-\epsilon)un\right)$. Now, note that given the form of $r_i$ we have,
\begin{equation*}
\begin{split}
r_{i} + g^i + \Delta^i &\leq \min\left(\sum^{i-1}_{j=1}g^j, n-\sum^{i}_{j=1}g^j-\Delta^i\right) + g^i + \Delta^i\\
\quad & \leq \min\left(\sum^{i-1}_{j=1}g^j, n-\sum^{i}_{j=1}g^j-\Delta^i\right) + 2\epsilon\min\left(\sum^{i-1}_{j=1}g^j, n-\sum^{i}_{j=1}g^j-\Delta^i\right)\\
\quad & \leq (1+2\epsilon)r_i.
\end{split}
\end{equation*}
So now combining this with (\ref{equation:greaterthanrcondition}) we have that $(1+2\epsilon)r_i > \min\left((1+\epsilon)un,\epsilon n + (1-\epsilon)un\right)$, and
$$
r_i > \min\left((1-\epsilon)un, (1+\epsilon)un - \epsilon n\right).
$$
So overall we have,
$$
\underbrace{\min \left((1-\epsilon)un, (1+\epsilon)un-\epsilon n\right)}_{\text{(a)}} < r_{i-1}+g^{i-1} \leq r_{i-1}+g^{i-1}+\Delta^{i-1} \leq \underbrace{\min\left((1+\epsilon)un,\epsilon n + (1-\epsilon)un\right)}_{\text{(b)}}.
$$
By the construction of $Q$ we know that the value $l$ that satisfies $\tilde{x}_{(k)}^l=v_i$, lies in between $r_i+g^i$ and $r_i+g^i+\Delta^i$. Therefore we need to show that (b) $-$ (a) $\leq 2\epsilon\min\left(un,(1-u)n\right)$. If $u \leq 0.5$, then $(1-\epsilon)un$ is the minimum of (a) and $(1+\epsilon)un$ is the minimum of (b), and therefore,
$$(1+\epsilon)un-(1-\epsilon)un = 2\epsilon un.$$
If $u > 0.5$, then $(1+\epsilon)un-\epsilon n$ is the minimum of (a) and $(1-\epsilon)un+\epsilon n$ is the minimum of (b), and therefore
$$
\left(\epsilon n + (1-\epsilon)un\right)-\left((1+\epsilon )un - \epsilon n\right) = 2\epsilon n - \epsilon un - \epsilon un = 2 \epsilon n (1- u).
$$
\qed
\end{appendix}

\bibliography{refs}
\end{document}